\documentclass[letterpaper,11pt]{article}
\usepackage{hyperref}
\usepackage{amssymb,amsmath,amssymb,amsthm}
\usepackage{setspace}
\usepackage[dvips]{graphicx}
\usepackage[letterpaper,nohead,left=0.795in,right=0.795in,top=0.95in,bottom=0.9in]{geometry}

\def\setof#1{\left\{{\let\st\colon #1 }\right\}}

  \def\calG{\mathcal{G}} 
\def\calC{\mathcal{C}}  \def\calS{\mathcal{S}}
\def\calP{\mathcal{P}} \def\calT{\mathcal{T}} \def\kk{\mathbf{k}} \def\ll{\boldsymbol{\ell}}
\def\gen{\texttt{gen}} \def\genb{\texttt{gen-block}} \def\genp{\texttt{gen-pair}}

\def\bdelta{\boldsymbol{\delta}} 

 \def\BOO={=^{\hspace{0.06cm}\epsilon}_B}
\def\bmu{\boldsymbol{\mu}} \def\bnu{\boldsymbol{\nu}}

  \def\ww{\mathbf{w}}
\def\00{\mathbf{0}}  
   
\def\uu{\mathbf{u}}  \def\xx{\mathbf{x}} \def\yy{\mathbf{y}}
\def\AA{\mathbf{A}} \def\BB{\mathbf{B}}  
\def\pp{\mathbf{p}}   \def\ss{\mathbf{s}}
 \def\xx{\mathbf{x}} \def\yy{\mathbf{y}} 
   
   \def\CC{\mathbf{C}}
\def\FF{\mathbf{F}}  \def\calG{\mathcal{G}}
 \def\calS{\mathcal{S}} \def\calP{\mathcal{P}}
\def\calC{\mathcal{C}}  
  \def\calV{\mathcal{V}}
  \def\bgamma{\boldsymbol{\gamma}}
\newtheorem{defi}{Definition}  \def\bmu{\boldsymbol{\mu}}
\newtheorem{lemm}{Lemma} 
\newtheorem{theo}{Theorem} 
\newtheorem{obse}{Observation}

  \def\frX{{\mathfrak{X}}}
\def\frY{{\mathfrak{Y}}}
\def\frS{{\mathfrak{X}}} \def\frT{{\mathfrak{Y}}} \def\11{\mathbf{1}}
 \def\DD{\mathbf{D}} \def\calE{\mathcal{E}}
\def\frR{\mathfrak{R}} \def\frP{\mathfrak{P}} \def\frQ{\mathfrak{Q}}
\def\balpha{\boldsymbol{\alpha}} \def\bbeta{\boldsymbol{\beta}}

\begin{document}

\title{\LARGE A Decidable Dichotomy Theorem on Directed\\
Graph Homomorphisms with Non-negative Weights\vspace{0.7cm}\thanks{We thank the following colleagues for their
    interest and helpful comments: Martin Dyer, Alan Frieze, Leslie Goldberg,\newline
    Richard Lipton, Pinyan Lu, and Leslie Valiant.}}
\author{Jin-Yi Cai\thanks{Computer Science Department, University of Wisconsin-Madison and Beijing University.
Research supported by NSF\newline grants
  CCF-0830488 and CCF-0914969.}
  \and
  Xi Chen\thanks{Computer Science Department, University of Southern California. Part of the research done while the author was\newline a postdoc at
  Princeton University. Supported by NSF grants CCF-0832797, DMS-0635607
  and a USC startup fund.}}
\maketitle \date \thispagestyle{empty}

\begin{spacing}{1.1}
\begin{abstract}
The complexity of graph homomorphism problems has been the subject of
intense study. It is a long standing open problem
to give a (decidable) complexity dichotomy theorem
for the partition function of directed graph homomorphisms.
In this paper, we prove a \emph{decidable} complexity dichotomy theorem
for this problem and our theorem applies to
all non-negative weighted form of the problem: given any fixed
matrix $\AA$ with
non-negative algebraic entries,
the  partition function $Z_\AA(G)$ of directed graph homomorphisms
from any directed graph $G$ is \emph{either} tractable in
polynomial time \emph{or} \#P-hard, depending on the matrix $\AA$.
The proof of the dichotomy theorem is combinatorial,
but involves the definition of an infi\-nite family
of graph homomorphism problems. The  proof of its decidability
is algebraic using properties of polynomials.
\end{abstract}
\end{spacing}
\newpage\setcounter{page}{1}

\section{Introduction}\vspace{-0.05cm}

The complexity of counting graph homomorphisms has received much attention recently
  \cite{DyerGreenhill,BulatovGrohe,BulatovD03,Bulatov08,acyclic,GoldbergGJT,CCL09}.
The problem can be defined for both
 \emph{directed} and
 \emph{undirected}   graphs.
Most results have been obtained for  \emph{undirected} graphs,
while the study of complexity of
the problem is significantly more challenging
for \emph{directed} graphs.
In particular, Feder and Vardi showed that the decision
  problems defined by directed graph homomorphisms are as general as
  the Constraint Satisfaction Problems (CSPs), and a complexity dichotomy
  for the former would resolve their long standing dichotomy conjecture for all
  CSPs \cite{FederVardi}.\vspace{0.01cm}

Let $G$ and $H$ be two graphs.
We follow the standard definition of graph homomorphisms, where
$G$ is allowed to have multiple edges but no self loops; and $H$ can
have both multiple edges and self loops.\hspace{0.05cm}\footnote{However, our results are actually
stronger in that our tractability result allows for loops in $G$, while
our hardness result holds for $G$ without loops.}
We say $\xi:V(G)\rightarrow V(H)$ is
  a graph homomorphism from $G$ to $H$ if $\xi(u)\xi(v)$ is an edge in $E(H)$
  for all $uv\in E(G)$.
Here if $H$ is an \emph{undirected} graph, then
 $G$ is also an undirected graph;
if  $H$ is \emph{directed}, then
 $G$ is also directed.
The undirected problem is a special case of the directed one.

For a fixed $H$, we are interested in the complexity of the following integer
  function \hspace{-0.05cm}$Z_H(G)$:
The input is a graph $G$, and the output is the number of graph homomorphisms
  from $G$ to $H$.
More generally, we can define $Z_\AA(\cdot)$ for any fixed
 $m\times m$ matrix $\AA=(A_{i,j})$:\vspace{-0.03cm}
$$
Z_\AA(G)=\sum_{\xi:V\rightarrow [m]} \hspace{0.07cm}\prod_{uv\in E} A_{\xi(u),\xi(v)},
  \ \ \ \ \text{for any directed graph $G=(V,E)$.}\vspace{-0.03cm}
$$
Note that the input $G$ is a directed graph in general.
However, if $\AA$ is a symmetric matrix, then one can always view
$G$ as an undirected graph.
Moreover, if $\AA$ is a $\{0,1\}$-matrix, then $Z_{\AA}(\cdot)$ is exactly $Z_{H}(\cdot)$,
  where $H$ is the graph whose adjacency matrix is $\AA$.\vspace{0.015cm}

Graph homomorphisms can express many interesting counting problems over graphs.
For example, if we take $H$ to be an undirected graph over two vertices $\{0,1\}$ with
  an edge $(0,1)$ and a loop $(1,1)$
 at $1$, then a graph homomorphism from $G$ to $H$ corresponds
  to a {\sc Vertex Cover} of $G$, and $Z_H(G)$ is simply the number of vertex covers of $G$.
As another example, if $H$ is the complete graph on $k$ vertices without self loops,
  then $Z_H(G)$ is the number of $k$-{\sc Colorings} of $G$.
In \cite{freedman}, Freedman, Lov\'{a}sz, and Schrijver characterized what graph
  functions can be expressed as $Z_\AA(\cdot)$.\vspace{0.015cm}


For increasingly more general families $\calC$ of matrices $\AA$, the
  complexity of $Z_{\AA}(\cdot)$
  has been studied and \emph{dichotomy} theorems have been proved.
A \emph{dichotomy} theorem for a given family $\calC$ of matrices $\AA$ states that
 for any $\AA\in \calC$, the problem of computing $Z_{\AA}(\cdot)$
  is either \emph{in polynomial time} or \emph{$\#P$-hard}.
A \emph{decidable} dichotomy theorem requires that
  the dichotomy criterion is computably decidable: There is a finite-time
  classification algorithm that, given any $\AA\in \calC$, decides whether
  $Z_\AA(\cdot)$ is in polynomial time or \#P-hard.
Most results have been obtained for undirected graphs.\vspace{0.3cm}

\noindent\textbf{Symmetric matrices $\AA$, and $Z_\AA(G)$ over \vspace{0.01cm}undirected graphs $G$:}\\
In \cite{Hell1, Hell2}, Hell and Ne\v{s}et\v{r}il showed that
  given any symmetric $\{0,1\}$ matrix $\AA$, \emph{deciding} whether $Z_\AA(G)$ $>0$
  is either in P or NP-complete.
Then Dyer and Greenhill \cite{DyerGreenhill}
showed that given
  any symmetric~$\{0,\hspace{-0.05cm}1\}$ matrix $\AA$, the problem of computing $Z_{\AA}(\cdot)$ is either in P or \#P-complete.
Bulatov and Grohe generalized their result to all non-negative
  symmetric matrices $\AA$ \cite{BulatovGrohe}.\footnote{More exactly, they proved a dichotomy
  theorem for symmetric matrices $\AA$ in
  which every entry $A_{i,j}$ is a non-negative algebraic number.
Our result in this paper applies similarly to all non-negative algebraic numbers,
  and throughout the paper we use $\mathbb{R}$ to denote the set of real algebraic numbers.}
They obtained an elegant dichotomy theorem which basically says that
  $Z_{\AA}(\cdot)$ is in P if~every~\emph{block} of $\AA$
  has rank at most one, and is \#P-hard otherwise.
In \cite{GoldbergGJT} Goldberg, Grohe, Jerrum and Thurley
proved a beautiful dichotomy for all
  symmetric real matrices.
Finally, a dichotomy theorem for all
  symmetric  complex matrices was recently proved by Cai, Chen and \newpage \noindent Lu \cite{CCL09}.
We remark that all these dichotomy theorems for symmetric
  matrices above are \emph{polynomial-time} decidable,
  meaning that given any matrix $\AA$, one can decide in polynomial time (in the input size of $\AA$)
   whether $Z_{\AA}(\cdot)$ is in P or \#P-hard.\vspace{0.3cm}

\noindent\textbf{General matrices $\AA$, and $Z_\AA(G)$ over directed\vspace{0.01cm} graphs $G$:}\\
In a paper that won the best paper award at ICALP in 2006,
Dyer, Goldberg and Paterson~\cite{acyclic}
proved a dichotomy theorem for directed graph homomorphism problems $Z_H(\cdot)$, but
restricted to {directed \emph{acyclic} graphs} $H$.
They~in\-troduced the concept of \emph{Lov\'{a}sz-goodness} and
proved that $Z_H(\cdot)$ is
  in P if the graph~$H$ is \emph{layered}\hspace{0.03cm}\footnote{A directed acyclic graph is \emph{layered}
  if one can partition its vertices into $k$ sets $V_1,\ldots,V_k$, for some
  $k\ge 1$, such that every edge goes from $V_i$ to $V_{i+1}$ for some $i:1\le i<k$.\vspace{0.03cm}}
   and \emph{Lov\'{a}sz-good},
and is \#P-hard otherwise.
  The property of Lov\'{a}sz-goodness turns out to be polynomial-time decidable.\vspace{0.01cm}

In \cite{Bulatov08}, Bulatov presented a sweeping dichotomy theorem for all counting Constraint Satisfaction
  Pro\-blems.
Recently Dyer and Richerby~\cite{private2}
  obtained an alternative proof.
The dichotomy theorem of Bulatov then implies a dichotomy
  for $Z_H(\cdot)$ over all directed graphs $H$.
However, it is rather unclear whether~this dichotomy theorem is {decidable} or not.
The criterion~\footnote{A dichotomy criterion is a well-defined mathematical property over
  the family of matrices $\AA$ being considered such that $Z_\AA(\cdot)$ is in P if $\AA$
  has this property; and is \#P-hard otherwise.\vspace{0.02cm}}
  requires one to check a condition
  on an infinitary object (see Appendix \ref{Appen-cri} for details).
This situation remains the same for the Dyer-Richerby proof in \cite{private2}.
The decidability of the dichotomy was then left
  as an open problem in \cite{BulatovECCC}.\vspace{0.35cm}

In this paper, we prove a dichotomy theorem for the family of all non-negative real matrices $\AA$.~We
  show that for every fixed $m\times m$ non-negative matrix $\AA$,
the problem of computing $Z_{\AA}(\cdot)$ is either
  in P or \#P-hard.
Moreover, our dichotomy criterion is \emph{decidable}:\hspace{-0.03cm}
we give a finite-time algorithm which,~given any non-negative matrix $\AA$,
  decides whether $Z_{\AA}(\cdot)$ is in P or \#P-hard.
In particular, for the family of $\{0,1\}$ matrices
  our result gives an alternative dichotomy criterion\hspace{0.01cm}\footnote{Both our
  dichotomy criterion (when specialized
  to the $\{0,1\}$ case) and the one
  of Bulatov characterize $\{0, 1\}$ matrices $\AA$ with $Z_\AA(\cdot)$
  in P and thus, they must be equivalent, i.e., $\AA$ satisfies our criterion if and only if it
  satisfies the one of Bulatov.
As a corollary, our result also implies a finite-time algorithm for checking
  the dichotomy criterion of Bulatov \cite{BulatovECCC} (and the
  version of Dyer and Richerby \cite{private2}) for the case of $\{0, 1\}$ matrices $\AA$.}
to that of Bulatov \cite{BulatovECCC} and Dyer-Richerby~\cite{private2},
  which is decidable.\vspace{0.02cm}

The main difficulty we encountered in obtaining the dichotomy theorem is due to
  the abundance of new intricate but
tractable cases, when moving from acyclic graphs to general directed graphs.
For example, $H$ does not have to be \emph{layered} for the problem $Z_H(\cdot)$ to be tractable
  (see Figure \ref{fig:figurenew} in Appendix \ref{app:figure} for an example).
Because of the generality of directed graphs,
  it seems impossible to have a simply stated criterion
  (e.g., Lovasz-goodness, as was used in the acyclic case \cite{acyclic}) which is both
  powerful enough to completely characterize all the tractable cases and also easy to check.
However, we manage to find a dichotomy criterion as well as a finite-time algorithm to
  decide whether $\AA$ satisfies it or not.\vspace{0.02cm}

In particular, the dichotomy theorem of Dyer, Goldberg and Paterson \cite{acyclic}
  for the acyclic case fits into our framework as follows.
In our dichotomy we start from $\AA$ and then define, in each round, a (possibly infinite)
  set of new matrices.
The size of the matrices defined in round $i+1$ is strictly smaller than that of round $i$
  (so there could be at most $m$ rounds).
The dichotomy then is that $Z_\AA(\cdot)$ is in P if and only if every \emph{block} of any
  matrix defined in the process above is of rank $1$ (see Section \ref{sec:intuition} and \ref{sec:sketch}
  for details).
For the special acyclic case treated by Dyer, Goldberg and Paterson \cite{acyclic},
  let $\AA$ be the adjacency matrix of $H$ which
  is acyclic and has $k$ layers,
then at most $k$ rounds are necessary to reach a conclusion about whether $Z_\AA(\cdot)=Z_H(\cdot)$
  is in P or \#P-hard.
However, when $H$ has $k$ layers but is not acyclic (i.e., there are edges from layer
  $k$ to layer $1$), deciding whether $Z_\AA(\cdot)$ is in P or \#P-hard becomes
  much harder in the sense that we might need $\gg k$ rounds to reach a conclusion.\newpage\vspace{0.02cm}




After we circulated a draft of this paper,
Goldberg informed us that she and coauthors~\cite{private1} found
a reduction from  \emph{weighted} counting CSP with \emph{non-negative
rational} weights to the 0-1 dichotomy theorem of Bulatov~\cite{BulatovECCC}.
However, the combined result still only
  works for non-negative rational weights
and more importantly, the dichotomy is not known to
be decidable.
%

\subsection{Intuition of the Dichotomy: Domain Reduction}\label{sec:intuition}

Let $\AA$ be the $m\times m$ non-negative matrix being considered,
  and $G=(V,E)$ be the input directed graph.
Before giving a more formal sketch of the proofs, we use a simple example to illustrate one of the most
  important ideas of this work: \emph{domain reduction}.\vspace{0.015cm}

For this purpose we also need to introduce the concept of \emph{labeled directed graphs}.
A labeled directed graph $\calG$ over domain $[m]= \{1,2, \ldots,m\}$ is a directed graph, in
  which every directed edge $e$ is labeled with an $m\times m$ matrix $\AA^{[e]}$; and
  every vertex $v$ is labeled with an $m$-dimensional vector $\ww^{[v]}$.
Then the partition function of $\calG$ is defined as
$$
Z(\calG)=\sum_{\xi:V\rightarrow [m]}\hspace{0.1cm}
  \prod_{v\in V} \ww^{[v]}_{\xi(v)} \prod_{uv\in E} \AA^{[uv]}_{\xi(u),\xi(v)}.
$$
In particular, we have $Z_\AA(G)=Z(\calG_0)$ where $\calG_0$ has the same graph structure as
  $G$; every edge of $\calG_0$ is labeled with
the same $\AA$; and every vertex of $\calG_0$
  is labeled with $\11$, the $m$-dimensional all-$1$ vector.\vspace{0.015cm}

Roughly speaking, starting from the input $G$,
  we build (in polynomial time) a \emph{finite} sequence of new \emph{labeled} directed graphs
  $\calG_0,\hspace{0.03cm}\calG_1,\hspace{0.03cm}\calG_2,\ldots,\calG_h$ one by one.
$\calG_{k+1}$ is constructed from $\calG_k$ by using the domain reduction method which we are
  going to describe next.
On the one hand, the \emph{domains} of these labeled graphs sh\-rink along with $k$. This means, the
  size of the edge weight matrices associated with the edges of $\calG_k$ (or equivalently,
  the dimension
 of the vectors associated with the vertices of $\calG_k$) strictly decreases
 along with $k$.
On the other hand, we have $Z(\calG_{k+1})=Z(\calG_{k})$ for all $k\ge 0$
 and thus,
  $$Z_\AA(G)=Z(\calG_0)=\ldots=Z(\calG_h).$$
Since the domain size decreases monotonically, the number of graphs $\calG_k$ in this sequence
  is at most $m$.
To prove our dichotomy theorem, we show that,
either something bad happens which forces us to stop the domain reduction
  process, in which case we show that $Z_\AA(\cdot)$ is \#P-hard;
  or we can keep reducing the domain size until the computation becomes trivial,
  in which case we show that $Z_\AA(\cdot)$ is in P.\vspace{0.015cm}

We say a matrix $\AA$ is
 \emph{block-rank-}$1$ if one can (separately) permute the rows and columns of $\AA$
  to get a~block diagonal matrix in which every block is of rank at most $1$.
If $\AA$ is not block-rank-$1$ we can easily show that $Z_\AA(\cdot)$ is \#P-hard, using the dichotomy of Bulatov
  and Grohe \cite{BulatovGrohe} for symmetric non-negative matrices (see Lemma \ref{lemma1}).
So without loss of generality, we assume $\AA$ is block-rank-1.
For example, let $\AA$ be the $8\times 8$ {block rank-$1$} non-negative matrix
  in Figure \ref{fig:figure1} in Appendix \ref{app:figure} with $16$ positive entries.
Then we use $\calT=\{(A_1,B_1),(A_2,B_2),(A_3,B_3),(A_4,B_4)\}$ to denote the \emph{block structure} of $\AA$, where
$$
\forall\hspace{0.04cm} s\in [4],\ A_s=\{2s-1,2s\},\ B_1=\{1,3\},\ B_2=\{5,7\},\ B_3=\{2,4\}\ \ \text{and}\ \
B_4=\{6,8\},
$$
so that $A_{i,j}>0$ if and only if $i\in A_s$ and $j\in B_s$, for some $s\in [4]$.
Because $\AA$ is block-rank-$1$, there also exist two $8$-dimensional
  \emph{positive} vectors $\balpha$ and $\bbeta$ such that
$$
A_{i,j}=\alpha_i\cdot \beta_j,\ \ \ \ \text{for all $(i,j)$ such that $i\in A_s$ and $j\in B_s$
  for some $s\in [4]$.}
$$

Now let $G=(V,E)$ be the directed graph in Figure \ref{fig:figure2}, where $|V|=6$
  and $|E\hspace{0.03cm}|=6$.
We illustrate the \emph{domain reduction}
  process by constructing the first labeled directed graph $\calG_1$ in the sequen\-ce as follows.
To simplify the presentation, we let $\yy\in [8]^6$ (instead of $\xi:V\rightarrow [8]$)
  denote an assignment, where $y_i\in [8]$ denotes the value of vertex $i$ in Figure \ref{fig:figure2} for every $i\in [6]$.\vspace{0.02cm}

First, let $\yy\in [8]^6$ be any assignment with a nonzero weight:
  $A_{y_i,y_j}>0$ for every edge $ij\in E$.
Since $\AA$ has the~block structure $\calT$, for every $ij\in E$, there exists a unique index $s\in [4]$ such that
  $y_i\in A_s$ and $y_j\in B_s$.
This inspires us to introduce a new variable $x_\ell\in [4]$ for each edge $e_\ell\in E$,
  $\ell\in [6]$ (as shown in Figure~\ref{fig:figure2}).
For every possible assignment of $\xx=(x_1,x_2,\ldots,x_6)\in [4]^6$, we use $Y[\xx]$ to denote
  the set of all possible assignments $\yy\in [8]^6$ such that
  for every $e_\ell=ij$, $y_i\in A_{x_\ell}$ and $y_j\in B_{x_\ell}$.
Now we have
$$
Z_\AA(G)=\sum_{\xx\in [4]^6}\hspace{0.08cm} \sum_{\yy\in Y[\xx]}
  \text{wt}(\yy),\ \ \ \ \text{where $\text{wt}(\yy)= \prod_{ij\in E} A_{y_i,y_j}$.}
$$

Second, we further simplify the sum above by noticing that if $x_2\ne x_3$ in $\xx$, then
  $Y[\xx]$ must be empty because the two edges $e_2$ and $e_3$
  share the same tail in $G$.
In general, we only need to sum over the case when $x_1=x_2=x_3$ and $x_4=x_5$,
  since otherwise the set $Y[\xx]$ is empty.
As a result,
$$
Z_\AA(G)=\sum_{\substack{x_1=x_2=x_3\\x_4=x_5\\x_6}}\hspace{0.08cm} \sum_{\yy\in Y[\xx]}
  \text{wt}(\yy).
$$

The advantage of introducing $x_\ell$, $\ell\in [6]$, is that, once $\xx$ is fixed,
  one can always decom\-pose $A_{y_i,y_j}$ as a product $\alpha_{y_i}\cdot \beta_{y_j}$,
  for all $\yy\in Y[\xx]$ and all $ij\in E$,
since $\yy$ belonging to $Y[\xx]$ guarantees that $(y_i,y_j)$ falls inside
  one of the four blocks of $\AA$.
This allows us to greatly simplify $\text{wt}(\yy)$: If $\yy\in Y[\xx]$, then\vspace{0.05cm}
$$
\text{wt}(\yy)= A_{y_1,y_3}\cdot A_{y_1,y_2}\cdot A_{y_2,y_3}\cdot
  A_{y_3,y_4}\cdot A_{y_3,y_5}\cdot A_{y_5,y_6}
  =\alpha_{y_1}\beta_{y_3}\alpha_{y_1}\beta_{y_2}\alpha_{y_2}\beta_{y_3}\alpha_{y_3}\beta_{y_4}
  \alpha_{y_3}\beta_{y_5}\alpha_{y_5}\beta_{y_6}.\vspace{0.07cm}
$$
Also notice that $Y[\xx]$, for any $\xx$, is a direct product of subsets of $[8]$:
$\yy\in Y[\xx]$ if and only if\vspace{0.05cm}
\begin{eqnarray*}
&y_1 \in L_1 =  A_{x_1},\  y_2 \in  L_2 = A_{x_3} \cap  B_{x_1} = A_{x_1}
\cap B_{x_1},\  y_3 \in L_3 = A_{x_4} \cap A_{x_5} \cap  B_{x_2} \cap
B_{x_3}
= A_{x_4} \cap  B_{x_1}&\\
&y_4 \in L_4 =  B_{x_4}, y_5 \in L_5 =  A_{x_6} \cap  B_{x_4} , y_6 \in L_6
=  B_{x_6}.\vspace{0.3cm}&
\end{eqnarray*}
As a result, $Z_\AA(G)$ becomes
\begin{equation}\label{demo1}
Z_\AA(G)=\sum_{x_1,x_4,x_6}\hspace{0.1cm}
\sum_{y_i\in L_i,\hspace{0.05cm}i\in [6]} \left((\alpha_{y_1})^2 \alpha_{y_2}\beta_{y_2}\right)\cdot
  \left((\alpha_{y_3})^2 (\beta_{y_3})^2\right)\cdot \beta_{y_4}\cdot
  (\alpha_{y_5}\beta_{y_5})\cdot \beta_{y_6}.
\end{equation}

Finally we construct the following \emph{labeled} directed graph $\calG_1$ over domain $[4]$.
There are three vertices $a,b$ and $c$, which correspond to $x_1,x_4$
  and $x_6$, respectively;
  and there are two directed edges $ab$ and $bc$.
We construct the weights as follows. The vertex weight vector of $a$ is
$$
w^{[a]}_\ell= \sum_{y_1\in A_\ell,\hspace{0.08cm}
  y_2\in A_\ell\cap B_{\ell}} (\alpha_{y_1})^2\cdot (\alpha_{y_2}\beta_{y_2}),
\ \ \ \ \text{for every $\ell\in [4]$};
$$
the vertex weights of $b$ and $c$ are the same:
\begin{eqnarray*}
w^{[b]}_\ell=w^{[c]}_\ell=\sum_{y \in B_{\ell}}\beta_{y },\ \ \ \ \text{for every $\ell\in [4]$.}
\end{eqnarray*}
The edge weight matrix $\CC^{[ab]}$ of $ab$ is
\begin{eqnarray*}
C^{[ab]}_{k,\ell}=\sum_{y_3\in B_k\cap A_\ell} (\alpha_{y_3})^2(\beta_{y_3})^2,
  \ \ \ \ \text{for all $k,\ell\in [4]$;}
\end{eqnarray*}
and the edge weight matrix $\CC^{[bc]}$ of $bc$ is
\begin{eqnarray*}
C^{[bc]}_{k,\ell}=\sum_{y_5\in B_k\cap A_\ell} \alpha_{y_5} \beta_{y_5},
\ \ \ \ \text{for all $k,\ell\in [4]$.}
\end{eqnarray*}\newpage
\noindent Using (\ref{demo1}) and the definition of
  $Z(\calG_1)$, it is easy to verify that $Z_\AA(G)=Z(\calG_1)$ and thus, we reduced the domain size
  of the problem from $8$ (which is
the number of rows and columns in $\AA$),
 to $4$ (which is the number of blocks in $\AA$).
However, we also paid a high price.
Two issues are worth pointing out here:
\begin{enumerate}
\item Unlike in $Z_\AA(G)$, different edges in $\calG_1$ have \emph{different} edge weight matrices in general.
For example, the matrices associated with $ab$ and $bc$ are clearly different,
  for general $\balpha$ and $\bbeta$.
Actually, the set of matrices that may appear as an edge weight of $\calG_1$, constructed from
{\it all possible} directed graphs $G$ after one round of domain reduction, is
  \emph{infinite} in general.\vspace{-0.1cm}

\item Unlike in $Z_\AA(G)$, we have to introduce vertex weights in $\calG_1$.
Similarly,
  vertices may have different vertex weight vectors, and the set of
  vectors that may appear as a vertex weight of $\calG_1$, constructed from
{\it all possible} $G$ after one round of domain reduction, is
  \emph{infinite} in general.\vspace{0.06cm}
\end{enumerate}
It is also worth noticing that even if the matrix $\AA$ we start with is $\{0,1\}$,
  the edge and vertex weights~of $\calG_1$ immediately become \emph{rational} right after the first round
  of domain reduction and we have to deal with rational weights afterwards.
So $\{0,1\}$-matrices are not that special under this framework.\vspace{0.015cm}

These two issues cause us a lot of trouble because we need to carry out the
  domain reduction process for several times, until the computation becomes trivial.
However, the reduction process above crucially used the assumption that
  $\AA$ is block-rank-$1$ (otherwise one cannot replace $A_{i,j}$ with $\alpha_i\cdot\beta_j$).
Therefore, there is no way to continue this process if
  some edge weight matrix in $\calG_1$ is \emph{not} block-rank-1.
To deal with this case, we show that if this happens for some $G$, then $Z_\AA(\cdot)$ is \#P-hard.
Informally, we have\vspace{0.038cm}

\begin{theo}[Informal]\label{informal1}
\hspace{-0.03cm}For any $G$, if one of the edge matrices in $\calG_k$ \emph{(constructed from $G$ after $k$ rounds of domain reductions)}, for some $k\ge 1$, is \emph{not} block-rank-$1$, then
  $Z_\AA(\cdot)$ is \emph{\#P-hard}.\vspace{0.038cm}
\end{theo}

The proof of Theorem \ref{informal1} for $k=1$ is relatively straight forward,
  because every edge weight matrix in $G$ is $\AA$.
However, due of the two issues mentioned earlier, the edge weights and vertex weights of $\calG_1$
  are drawn from infinite sets  in general, and even proving it for $k=2$
  is highly non-trivial.\vspace{0.015cm}

Even with Theorem \ref{informal1} which essentially gives us a dichotomy theorem for
  all non-negative matrices, it is still unclear whether the dichotomy is \emph{decidable} or not.
The difficulty is that, to decide whether $Z_\AA(\cdot)$ is in P or \#P-hard,
  we need to check infinitely
 many matrices (all the edge weight matrices that
  appear in the domain reduction process, from \emph{all possible} directed graphs $G$)
  and to see whether all of them are block-rank-$1$.
To overcome this, we give an algebraic proof using properties of
polynomials.  We manage to show that it is not necessary to check these matrices one by one,
  but only need to check whether or not the entries of $\AA$ satisfy finitely many polynomial constraints.

\subsection{Proof Sketch}\label{sec:sketch}

Without loss of generality, we assume that $\AA$ is a nonnegative block-rank-$1$ matrix.
To show that $Z_\AA(\cdot)$ is either in P or \#P-hard, we use the following two steps.\vspace{0.013cm}

In the first step, we \emph{define} from $\AA$ a finite sequence of pairs:
$$
(\frS_0,\frY_0),(\frS_1,\frY_1),\ldots,(\frS_h,\frY_h),\ \ \ \ \text{for some $h:0\le h<m$,}$$
where $\frX_0=\{\11\}$, $\frY_0=\{\AA\}$ and
  $\11$ denotes the $m$-dimensional all-$1$ vector.
Each pair $(\frS_k,\frT_k)$, $k\in [h]$, is defined from $(\frS_{k-1},\frT_{k-1})$.
Roughly speaking, $\frY_k$ (resp. $\frX_k$) is the set of all edge matrices
  (resp. vertex vectors) that may appear in $\calG_k$, after $k$ rounds of
  domain reductions.
There also exist positive integers\vspace{-0.04cm}
$$
m=m_0>m_1>\ldots>m_h\ge 1\vspace{-0.04cm}
$$
such that every $\frY_k$, $k\in [h]$, is a set of $m_k\times m_k$ non-negative
  matrices;
  and every $\frX_k$, $k\in [h]$, is a set of $m_k$-dimensional non-negative vectors.
Although the sets $\frX_k$ and $\frY_k$ are \emph{infinite} in general (which is
  the reason why we used the word ``\emph{define}'' instead of ``\emph{construct}''),
  the definition of $(\frS_k,\frY_k)$ guarantees the following two properties:\newpage
\begin{enumerate}
\item For each $k\in [h]$, all matrices in
  $\frY_k$ share the same \emph{structure}: $\forall\hspace{0.06cm}\BB,\BB'\in \frY_k$,
  $B_{i,j}>0\hspace{0.05cm} \Leftrightarrow\hspace{0.05cm} B_{i,j}'>0$;\vspace{-0.07cm}

\item Every matrix $\BB$ in $\frY_h$ is a \emph{permutation} matrix.
\end{enumerate}

The definition of $(\frS_k,\frY_k)$ from $(\frX_{k-1},\frY_{k-1})$ can be found in Appendix \ref{construction}.
In Appendix \ref{reduction} we prove that for every $k\in [h]$, if $\BB\in \frY_k$,
  then the problem of computing $Z_{\BB}(\cdot)$ is {polynomial-time reducible}
  to the~com\-putation of $Z_{\AA}(\cdot)$.
From this, we can obtain the hardness part of our dichotomy theorem:
  If for some $k\in [h]$, there exists a matrix
  $\BB\in \frY_k$ such that $\BB$ is not block-rank-1, then $Z_\AA(\cdot)$
  is \#P-hard.\vspace{0.013cm}

Now we assume that all matrices in $\frY_k$, $k\in [h]$, are block-rank-1.
To finish the proof we only need to show that if this is true,
  then $Z_\AA(\cdot)$ is indeed in P.
To this end, we use the \emph{domain reduction} process to
  construct a sequence of \emph{labeled} directed graphs
  $\calG_1,\ldots,\calG_h$ such that
\begin{enumerate}
\item $Z(\calG_1)=Z_\AA(G)$ and $Z(\calG_{k+1})=Z(\calG_k)$ for all $k:1\le k<h$; and\vspace{-0.07cm}
\item For every $k\in [h]$, we have $\AA^{[e]}\in \frY_k$ for all edges $e$ in $\calG_k$
  and $\ww^{[v]}\in \frX_k$ for all vertices $v$ in $\calG_k$.
\end{enumerate}
This sequence can be constructed in polynomial time, because the
  construction of ${\cal G}_{k+1}$ from ${\cal G}_k$ can be done very efficiently
  as described in Section \ref{sec:intuition}, and also because the number of graphs in
  the sequence is at most $m$.
By the two properties above, we have $Z_\AA(G)=Z(\calG_h)$; and every edge weight matrix $\AA^{[e]}$ in $\calG_h$
  is a \emph{permutation} matrix.
As a result, we can compute $Z_\AA(G)$ in polynomial time since
  $Z(\calG_h)$ can be computed efficiently.\vspace{0.013cm}

This finishes the proof of our dichotomy theorem:
  \hspace{-0.05cm}given any non-negative matrix $\AA$,
  the problem of computing $Z_\AA(\cdot)$ is either in polynomial time or \#P-hard.
Moreover, to decide which case it is, we only need to check whether the matrices in
  $\frY_k$, $k\in [h]$, satisfy the following condition:
\begin{quote}
\emph{\textbf{\emph{The Block-Rank-1 Condition}}}: Every matrix $\BB\in \frY_k$, $k\in [h]$, is block-rank-$1$.
\end{quote}
However, as mentioned earlier, all the sets $\frY_k$, $k\in [h]$, are infinite in general, so
  one cannot check the matrices one by one.
Instead, we express the block-rank-$1$ condition as a finite collection of
  polynomial constraints over $\frY_k$.
The way $(\frX_k,\frY_k)$ is defined from $(\frX_{k-1},\frY_{k-1})$ allows us to
  prove that, to check whether every matrix in $\frY_{k}$ (or every vector in $\frX_k$)
  satisfies a certain polynomial constraint, one only needs to check a finitely
  many polynomial constraints for $(\frX_{k-1},\frY_{k-1})$.
Therefore, to check whether $\frY_{k}$, $k\in [h]$, satisfies the block-rank-$1$
  condition we only need to check a finitely many polynomial constraints
  for $(\frX_0,\frY_0)$.
Since $\frX_0=\{\11\}$ and $\frY_0=\{\AA\}$ are both finite,
  this can be done in a finite number of steps.\vspace{-0.04cm}

\section{Preliminaries}\label{prelim}

We say $\calG=(G,\calV,\calE)$ is a \emph{labeled directed graph} over $[m]=\{1,\ldots,m\}$
  for some positive integer $m$, if
\begin{enumerate}
\item $G=(V,E)$ is a directed graph (which may have parallel edges but no self-loops);
\item Every vertex $v\in V$ is labeled with an $m$-dimensional non-negative vector
  $\calV(v)\in \mathbb{R}_+^{m}$ as its\\ vertex weight; and
\item Every edge $uv\in E$ is labeled with an $m\times m$ (not necessarily
  symmetric) non-negative matrix\\ $\calE(uv)\in \mathbb{R}_+^{m\times m}$
  as its edge weight.
\end{enumerate}

Let $\calG=(G,\calV,\calE)$ be a labeled directed graph, where $G=(V,E)$.
For each $v\in V$, we use $\ww^{[v]}=\calV(v)$ to denote its vertex weight vector;
  and for each $uv\in E$, we use $\CC^{[uv]}=\calE(uv)$ to denote its edge weight matrix.
Then we define $Z(\calG)$ as follows:
$$
Z(\calG)=\sum_{\xi:V\rightarrow [m]} \text{wt}(\calG,\xi),\ \ \ \
\text{where \hspace{0.08cm}$\text{wt}(\calG,\xi)=\prod_{v\in V}\hspace{0.05cm} w^{[v]}_{\xi(v)}\hspace{0.05cm}\prod_{uv\in E}
  C^{[uv]}_{\xi(u),\hspace{0.05cm}\xi(v)}$}
$$
denotes the \emph{weight} of the assignment $\xi$.\vspace{0.015cm}

Let $\CC$ be an $m\times m$ non-negative matrix.
We are interested in the complexity of $Z_\CC(\cdot)$:
$$
Z_\CC(G)=Z(\calG),\ \ \ \ \ \text{for any directed graph $G=(V,E)$,}
$$
where $\calG=(G,\calV,\calE)$ is the labeled directed graph
  with $\calV(v)=\11\in \mathbb{R}_+^m$ for all $v\in V$ and
  $\calE(uv)=\CC$ for all edges $uv\in E$.

\begin{defi}[Pattern and block pattern]
We say $\calP$ is an $m\times m$ \emph{pattern} if $\calP\subseteq [m]\hspace{-0.04cm}\times
  \hspace{-0.04cm} [m]$.
$\calP$ is said to be \emph{trivial} if $\calP=\emptyset$.
A non-negative $m\times m$ matrix $\CC$ is of \emph{pattern} $\calP$, if for all $i,j\in [m]$, we have
  $C_{i,j}$ $>0$ if and only if $(i,j)\in \calP$.
$\CC$ is also called a \emph{$\calP$-matrix}.
We say $\calT$ is an $m\times m$ \emph{block pattern} if
\begin{enumerate}
\item $\calT=\big\{(A_1,B_1),\ldots,(A_r,B_r)\big\}$
  for some $r\ge 0$;\vspace{-0.07cm}
\item $A_i\subseteq [m]$, $A_i\ne \emptyset$, $B_i\subseteq [m]$ and $B_i\ne \emptyset$ for
  all $i\in [r]$; and\vspace{-0.07cm}
\item $A_i\cap A_j=B_i\cap B_j=\emptyset$, for all $i\ne j\in [r]$.
\end{enumerate}
$\calT$ is said to be \emph{trivial} if $\calT=\emptyset$.
A \emph{block pattern} $\calT$ naturally defines a \emph{pattern} $\calP$, where\vspace{-0.04cm}
$$
\calP = \big\{\hspace{0.03cm}(i,j) \hspace{0.08cm} \big| \hspace{0.07cm}
  \exists\hspace{0.06cm}k \in [r]\ \text{such that \hspace{0.04cm}$i \in A_k$ and $j \in B_k$}\big\}.
$$
We also say $\calP$ is \emph{consistent} with $\calT$.
Finally, we say a non-negative $m \times m$ matrix $\CC$ is of \emph{block pattern} $\calT$,
  if $\CC$ is of pattern $\calP$ defined by $\calT$.
$\CC$ is also called a \emph{$\calT$-matrix}.\vspace{0.05cm}
\end{defi}

\begin{defi}
We say an $m\times m$ non-negative matrix $\CC$ is \emph{block-rank-$1$} if
\begin{enumerate}
\item Either $\CC=\mathbf{0}$ is the zero matrix \emph{(}and is of block pattern $\calT=\emptyset$\emph{)}; or\vspace{-0.07cm}
\item $\CC$ is of block pattern $\calT$, for some $m\times m$ block pattern $\calT=\{(A_1,B_1),
  \ldots,(A_r,B_r)\}$ with $r\ge 1$;\\ and for every $k\in [r]$,
  the sub-matrix of $\CC$ induced by $A_k$ and $B_k$ is \emph{(}exactly\emph{)} rank $1$.
\end{enumerate}
Let $\CC$ be a non-negative \emph{block-rank-$1$} matrix of block
  pattern $\calT$.
Then there exists a \emph{unique} pair $(\balpha,\bbeta)$ of
  non-negative $m$-dimensional vectors such that
\begin{enumerate}
\item For every $i\in [m]$, $\alpha_i>0\hspace{0.04cm} \Longleftrightarrow\hspace{0.04cm}
  i\in \bigcup_{k\in [r]} A_k$;
  and $\beta_i>0\hspace{0.04cm} \Longleftrightarrow\hspace{0.04cm} i\in \bigcup_{k\in [r]} B_k$;\vspace{-0.07cm}
\item $C_{i,j}=\alpha_i\cdot \beta_j$ for all $i,j\in [m]$ such that $C_{i,j}>0$; and\vspace{-0.07cm}
\item $\sum_{j\in A_i} \alpha_j=1$, for all $i\in [r]$.
\end{enumerate}
The pair $(\balpha,\bbeta)$ is called the \emph{{(}vector{)} representation} of $\CC$.
Note that we have $\balpha=\bbeta=\00$ when $\CC=\00$.\vspace{0.05cm}
\end{defi}

It is clear that $\calT$ and $(\balpha,\bbeta)$ together uniquely determine a
  non-negative block-rank-$1$ matrix.\vspace{0.006cm}

The following lemma concerns the complexity of $Z_\CC(\cdot)$.
The proof can be found in Appendix \ref{AppBulatovG}.\vspace{0.05cm}

\begin{lemm}\label{lemma1}
If $\CC$ is not block-rank-$1$, then $Z_\CC(\cdot)$ is \#P-hard.
\end{lemm}

Let $\calT$ be an $m\times m$ \emph{non-trivial} block pattern where
  $\calT=\{(A_1,B_1),\ldots,(A_r,B_r)\}$ for some $r\ge 1$.
It defines the following $r\times r$ pattern $\calP=\text{\gen}(\calT)$:
  For all $i,j\in [r]$,
$(i,j)\in \calP$ if and only if $B_i\cap A_j\ne \emptyset.$\vspace{0.018cm}

We also define $\text{\genb}(\calT)$ as follows:
\begin{enumerate}
\item If $\calP=\text{\gen}(\calT)$ is consistent with a block pattern, denoted by $\calT'$,
  then $\text{\genb}(\calT)=\calT'$;\vspace{-0.08cm}
\item Otherwise, we set $\text{\genb}(\calT)=\text{\rm{\emph{false}}}$.
\end{enumerate}
We note that $\calP=\text{\gen}(\calT)$ could be \emph{trivial} even if $\calT$ is \emph{non-trivial}.\vspace{0.012cm}


Next, we introduce a generalized version of $Z_\CC(\cdot)$.
Let $m\ge 1$ and $(\mathfrak{P},\mathfrak{Q})$ be a pair in which
\begin{enumerate}
\item $\frP$ is a \emph{finite} and nonempty set of non-negative $m$-dimensional vectors with $\11\in \frP$; and\vspace{-0.08cm}
\item $\frQ$ is a \emph{finite} and nonempty set of $m\times m$ non-negative matrices.
\end{enumerate}
We then use $Z(\cdot)$ to define the function $Z_{\frP,\frQ}(\cdot)$ as follows:
$$
Z_{\frP,\frQ}(\calG)=Z(\calG),
$$
where $\calG=(G,\calV,\calE)$ is a labeled directed graph with
  $\calV(v)\in \frP$ for any vertex $v\in V(G)$; and
  $\calE(uv)\in \frQ$ for any edge $uv\in E(G)$.
As an example, $Z_\CC(\cdot)$ is exactly $Z_{\frP,\frQ}(\cdot)$ with
  $\frP=\{\11\}$ and $\frQ=\{\CC\}$.\vspace{0.025cm}

Finally, let $m\ge 1$ and $(\frS,\frT)$ and $(\frS',\frT')$ be two pairs such that:
\begin{enumerate}
\item $\frS$ and $\frS'$ are two nonempty (and possibly infinite) sets of non-negative
  $m$-dimensional\\ vectors with $\mathbf{1}\in \frS$ and $\11\in \frS'$; and\vspace{-0.08cm}
\item $\frT$ and $\frT'$ are two nonempty (and possibly infinite)
  sets of non-negative $m\times m$ matrices.\vspace{0.06cm}
\end{enumerate}

\begin{defi}[Reduction]
We say $(\frS',\frT')$ is \emph{polynomial-time reducible} to $(\frS,\frT)$
  if for every finite and nonempty~subset $\frP'\subseteq \frS'$ with $\11\in \frP'$ and every
  finite and nonempty subset $\frQ'\subseteq \frT'$, there exist
  a finite and nonempty subset $\frP\subseteq \frS$ with $\11\in \frP$ and a finite
  and nonempty subset $\frQ\subseteq \frT$, such that $Z_{\frP',\frQ'}(\cdot)$ is \emph{polynomial-time reducible} to
  $Z_{\frP,\frQ}(\cdot)$.\vspace{-0.1cm}
\end{defi}

\section{Main Theorems}\label{mainstatement}\vspace{-0.02cm}

We prove a complexity dichotomy theorem for all counting problems $Z_\CC(\cdot)$ where $\CC$
  is any non-negative matrix.
Actually, our main theorem is more general.\vspace{0.03cm}

\begin{defi}
Let $\calP$ be an $m\times m$ pattern.
An $m$-dimensional non-negative vector $\ww$ is said to be\vspace{-0.03cm}
\begin{itemize}
\item[--] \emph{positive:} $w_i>0$ for all $i\in [m]$; and\vspace{-0.1cm}
\item[--] \emph{$\calP$-weakly positive:} for all $i\in [m]$, $w_i>0$ if and only if $(i,i)\in \calP$.
\end{itemize}
We call $(\frS,\frT)$ a \emph{$\calP$-pair} if
\begin{enumerate}
\item $\frS$ is a nonempty \emph{(}and possibly infinite\emph{)} set of
  \emph{positive} and $\calP$\emph{-weakly positive} vectors with $\11\in \frS$;\vspace{-0.1cm}
\item $\frT$ is a nonempty \emph{(}and possibly infinite\emph{)}
  set of $m\times m$ \emph{(}non-negative\emph{)} $\calP$-matrices.
\end{enumerate}
We say it is a \emph{finite} $\calP$-pair if both sets are finite.
We normally use $(\frP,\frQ)$ to denote a finite $\calP$-pair.\vspace{0.005cm}

Similarly, for any $m\times m$ block pattern $\calT$, we can define
  $\calT$\emph{-weakly positive vectors} as well as \emph{$\calT$-pairs} by replacing the $\calP$ above
  with the pattern defined by $\calT$.\vspace{0.03cm}
\end{defi}

We prove the following complexity dichotomy theorem:\vspace{0.02cm}

\begin{theo}[Complexity Dichotomy]\label{maintheorem}
Let $\calP$ be an $m\times m$ pattern for some $m\ge 1$,
  then for any \emph{finite $\calP$-pair} $(\frP,\frQ)$,
  the problem of computing $Z_{\frP,\frQ}(\cdot)$
  is either in polynomial time or \#P-hard.\vspace{0.05cm}
\end{theo}

Clearly, it gives us a dichotomy for the special case
  of $Z_{\CC}(\cdot)$ when $\frP=\{\11\}$ and $\frQ=\{\CC\}$.
Moreover, we show that for the special case when $\frP=\{\11\}$,
  we can decide in a finite number of steps whether $Z_{\frP,\frQ}$ is
   in polynomial time or \#P-hard.
In particular, it implies that the dichotomy for $Z_{\CC}(\cdot)$ is decidable.\vspace{0.05cm}

\begin{theo}[Decidability]\label{maindecidability}
Given any positive integer $m\ge 1$, an
  $m\times m$ pattern $\calP$, and a finite $\calP$-pair $(\frP,\frQ)$ with $\frP=\{\11\}$,
  the problem of whether $Z_{\frP,\frQ}(\cdot)$ is in polynomial time or \#P-hard is decidable.\vspace{0.05cm}
\end{theo}

We prove Theorem \ref{maintheorem} and \ref{maindecidability} in the rest of the section.
The lemmas (Lemma \ref{reductionlemm},
  \ref{tractabilitylemma}, and \ref{decidabilitylemma}) used in the proof will be proved in the appendix.

\subsection{Defining New Pairs: \genp\hspace{0.06cm}$(\frS,\frT)$}

Before proving Theorem \ref{maintheorem}, we state a key lemma
  which will be proved in Appendix \ref{construction} and Appendix \ref{reduction}.\vspace{0.01cm}

Let $(\frS,\frT)$ be a (possibly infinite) $\calT$-pair, for some non-trivial $m\times m$
  block pattern $\calT$.
Also assume that every matrix in $\frT$ is block-rank-$1$.
Then in Appendix \ref{construction}, we introduce an operation
  \genp\ over $(\frS,\frT)$, which \emph{defines} a new (and possibly infinite) pair
$
(\frS',\frT')=\text{\genp}(\frS,\frT).
$

\begin{defi}
A set $S$ of non-negative $m$-dimensional vectors, for some $m\ge 1$,
  is \emph{closed} if $\ww_1\circ \ww_2\in S$ for all vectors $\ww_1,\ww_2\in S$, where
  we let $\circ$ denote the \emph{Hadamard product} of two vectors: $\ww_1\circ \ww_2$
  is the $m$-dimensional vector whose $i$th entry is $w_{1,i}\cdot w_{2,i}$ for all $i\in [m]$.
\end{defi}

In Appendix \ref{reduction}, we prove the following lemma.

\begin{lemm}\label{reductionlemm}
Let $(\frS,\frT)$ be a $\calT$-pair, for some non-trivial
  block pattern $\calT$.
Suppose every matrix in $\frT$ is block-rank-$1$,
  then $(\frS',\frT')=\text{\emph{\genp}}(\frS,\frT)$ is a $\calP'$-pair,
  where $\calP'=\text{\emph{\gen}}(\calT)$.
The new vector set $\frS'$ is \emph{closed} and
  $(\frS',\frT')$ is \emph{polynomial-time reducible} to $(\frS,\frT)$.\vspace{-0.03cm}
\end{lemm}

\subsection{Proof of Theorem \ref{maintheorem}}\label{hahahasss}

Let $(\frP,\frQ)$ be a \emph{finite} $\calP$-pair, where $\calP$ is an $m\times m$ pattern.\vspace{0.01cm}

We assume $Z_{\frP,\frQ}(\cdot)$ is not \#P-hard, and we only
  need to show that $Z_{\frP,\frQ}(\cdot)$ is in polynomial time.\vspace{0.01cm}

By Lemma \ref{lemma1}, there must be a block pattern $\calT$ consistent with $\calP$
  and all the matrices in $\frQ$ are block- rank-$1$ since
  otherwise $Z_{\frP,\frQ}(\cdot)$ is \#P-hard, which contradicts the assumption.
Therefore, we have\vspace{0.06cm}
\begin{itemize}
\item[\textbf{R$_0$:}] $(\frP,\frQ)$ is a finite $\calT$-pair for some $m\times m$
  block pattern $\calT$; and\\ Every matrix in $\frQ$ is block-rank-$1$.\vspace{0.06cm}
\end{itemize}
For convenience, we rename $(\frP,\frQ)$ to be $(\frS_0,\frT_0)$
  and rename $m$ and $\calT$ to be $m_0$ and $\calT_0$, respectively.

Now we define a finite sequence of pairs using the \genp\ operation,
  starting with $(\frS_0,\frT_0)$.\vspace{0.005cm}

First, if $|A_i| =|B_i|=1$ for all $i$, i.e., every set $A_i$ and $B_i$ in
  $\calT_0$ is a singleton, then the sequence has only one pair
  $(\frS_0,\frT_0)$, and the definition of this sequence is complete.
Note that this also includes the special case when $\calT_0=\emptyset$
  and $\frT_0=\{\00\}$.

Otherwise, in Step $1$, we define a new $\calP_1$-pair $(\frS_1,\frT_1)$
  using \genp:\vspace{-0.03cm}
$$
(\frS_1,\frT_1)=\text{\genp}(\frS_0,\frT_0),\ \ \ \ \ \text{where
  $\calP_1=\text{\gen}(\calT_0)$.}\vspace{-0.03cm}
$$
By Lemma \ref{reductionlemm} $(\frS_1,\frT_1)$ is polynomial-time reducible to
  $(\frS_0,\frT_0)$.
This implies that $\calP_1$ must be consistent with a block pattern, denoted by $\calT_1$, and
  every matrix in $\frT_1$ is block-rank-$1$.
  (Otherwise, assume $\DD\in \frT_1$ is not block-rank-1, then
   by Lemma \ref{lemma1}, $Z_{\frP_1,\frQ_1}(\cdot)$ is \#P-hard, where
   $\frP_1=\{\11\}$ and $\frQ_1=\{\DD\}$.
It follows from Lemma \ref{reductionlemm} that there exists a finite pair
  $(\frP_0,\frQ_0)$ where $\frP_0\subseteq \frS_0$ and $\frQ_0\subseteq \frT_0$,
  such that $Z_{\frP_1,\frQ_1}(\cdot)$ is polynomial-time reducible to
  $Z_{\frP_0,\frQ_0}(\cdot)$.
On the other hand, it is clear that $Z_{\frP_0,\frQ_0}(\cdot)$ is
  reducible to $Z_{\frS_0,\frT_0}(\cdot)$ and thus, the latter is also \#P-hard, which contradicts our assumption.)
As a result, we have\vspace{0.12cm}
\begin{itemize}
\item[\textbf{R$_1$:}] $\calT_1=\text{\genb}(\calT_0)$ is an $m_1\times m_1$ block pattern,
  where $m_1$ is the number of pairs in $\calT_0$; \\
$(\frS_1,\frT_1)=\text{\genp}(\frS_0,\frT_0)$ is a $\calT_1$-pair, and
  every matrix in $\frT_1$ is block-rank-$1$.\vspace{0.12cm}
\end{itemize}
We also have $m_0>m_1$ since at least one of the sets in $\calT_0$
  is not a singleton.

We remark that both sets $\frS_1$ and $\frT_1$ are generally infinite, so
  one can not check the matrices in $\frT_1$ for the block-rank-$1$ property one by one.
It does not matter right now because we are only proving the dichotomy theorem.
However, it will become a serious problem later when we show that the dichotomy is decidable.
We have to show that the block-rank-$1$ property can be verified in a finite number of steps.

We then repeat the process above.
After $\ell\ge 1$ steps, we get a sequence of $\ell+1$ pairs:\vspace{-0.01cm}
$$
(\frS_0,\frT_0),(\frS_1,\frT_1),\ldots,(\frS_\ell,\frT_\ell),\vspace{-0.01cm}
$$
and $\ell+1$ block patterns $\calT_0,\calT_1,\ldots,\calT_\ell$ such that\vspace{0.06cm}
\begin{itemize}
\item[\textbf{R$_\ell$:}] For every $i\in [\ell]$, $\calT_i=\text{\genb}(\calT_{i-1})$; \\
For every $i\in [\ell]$, $(\frS_i,\frT_i)=\text{\genp}(\frS_{i-1},\frT_{i-1})$
  is a $\calT_i$-pair; and\\ For every $i\in [0:\ell]$, all the matrices in $\frT_i$
  are block-rank-$1$.\vspace{0.06cm}
\end{itemize}
We have two cases.
If every set in $\calT_\ell$ is a singleton (including the
  case when $\calT_\ell=\emptyset$ and $\frT_\ell=\{\00\}$), then the sequence has only $\ell+1$ pairs
  and the definition of the sequence is complete.
Otherwise in Step $\ell+1$ we apply the $\text{\genp}$ operation again to define
  a new pair $(\frS_{\ell+1},\frT_{\ell+1})$ from $(\frS_\ell,\frT_\ell)$.\vspace{0.015cm}

Finally, assuming $Z_{\frP,\frQ}(\cdot)$ is not \#P-hard,
  we get a sequence of $h+1$ pairs\vspace{-0.01cm}
$$
(\frS_0,\frT_0),(\frS_1,\frT_1),\ldots,(\frS_h,\frT_h),\ \ \ \ \ \text{for some $h\ge 0$,}\vspace{-0.01cm}
$$
together with $h+1$ positive integers $m_0>\ldots>m_h\ge 1$ and
  $h+1$ block patterns $\calT_0,\ldots,\calT_h$ such that\vspace{0.06cm}
\begin{itemize}
\item[\textbf{R:}]
For every $i\in [0:h]$, $\calT_i$ is an $m_i\times m_i$ block pattern;\\
For every $i\in [h]$, $\calT_i=\text{\genb}(\calT_{i-1})$;\\
Either $\calT_h=\emptyset$ is trivial or every set in $\calT_h$ is a singleton; \\
For every $i\in [h]$, $(\frS_i,\frT_i)=\text{\genp}
  (\frS_{i-1},\frT_{i-1})$ is a $\calT_i$-pair; and \\
For every $i\in [0:h]$, all the matrices in $\frT_i$ are block-rank-$1$.\vspace{0.06cm}
\end{itemize}
Because $m_0>\ldots>m_h\ge 1$, we also have $h<m_0=m$.\vspace{-0.04cm}

\subsubsection{Dichotomy}

Now we know that if $Z_{\frP,\frQ}(\cdot)$ is not \#P-hard, then
  there is a sequence of $h+1$ pairs for some $h:0\le h$ $<m$, which
  satisfies condition \textbf{(R)}.
To complete the dichotomy theorem, we show in Appendix \ref{tractability} that\vspace{0.05cm}

\begin{lemm}[Tractability]\label{tractabilitylemma}
\hspace{-0.06cm}Given any block pattern $\calT$ and a finite $\calT$-pair $(\frP,\frQ)$,
  let $(\frS_0,\frT_0),\ldots,(\frS_h,\frT_h)$ be a sequence of pairs defined as above,
  with $(\frS_0, \frT_0) = (\frP, \frQ)$.
Suppose it satisfies condition \emph{\textbf{(R)}},
then $Z_{\frP,\frQ}(\cdot)$ is computable in polynomial time.\vspace{0.04cm}
\end{lemm}

This finishes the proof of Theorem \ref{maintheorem}.\vspace{-0.01cm}

\subsection{Proof of Theorem \ref{maindecidability}}

Next, we show that for the special case when $\frS_0=\frP=\{\11\}$, the dichotomy theorem is decidable.\vspace{0.003cm}

First, the condition \textbf{(R$_{0}$)} can be checked easily since there
  are only finitely many matrices in $\frT_0$.\vspace{0.006cm}

Assume after $\ell: 0\le \ell<m$ steps, we get a sequence of $\ell+1$ pairs:
$
(\frS_0,\frT_0),(\frS_1,\frT_1),\ldots,(\frS_\ell,\frT_\ell),
$
together with $\ell+1$ block patterns $\calT_0,\ldots,\calT_\ell$.
Moreover, we know that they satisfy \textbf{(R$_{\ell}$)}.
If every set in $\calT_\ell$ is a singleton (including the case when $\calT_\ell=\emptyset$),
  then we are done because by Lemma \ref{tractabilitylemma}, the problem is in polynomial time.
Otherwise, to prove Theorem \ref{maindecidability}, we need a finite-time algorithm
  to check whether every matrix in the new $\calP$-pair
$(\frS_{\ell+1},\frT_{\ell+1})=\text{\genp}(\frS_\ell,\frT_\ell)$,
  {where $\calP=\gen(\calT_\ell)$,}
is block-rank-$1$ or not. We refer to this property as the \emph{rank property} for $\frT_{\ell+1}$.\vspace{0.01cm}

We prove the following lemma in Appendix \ref{decidability}.
Theorem \ref{maindecidability} then follows.\vspace{0.045cm}

\begin{lemm}\label{decidabilitylemma}
\hspace{-0.06cm}Given any block pattern $\calT$ and a finite $\calT$-pair $(\frS_0,\frT_0)$
  with $\frS_0=\{\11\}$, let $(\frS_0,\frT_0),\ldots,(\frS_\ell,\frT_\ell)$ be a
  sequence of pairs defined as above.
Suppose it satisfies condition \emph{(\textbf{R$_{\ell}$})}.
Then the \emph{rank property} for $\frT_{\ell+1}$ can be checked in a finite number of steps.
\end{lemm}

\newpage
\bibliographystyle{plain}
\begin{flushleft}
\bibliography{Reference}
\end{flushleft}
\newpage
\appendix

\section{Figures}\label{app:figure}

\begin{figure}[h!]\centering
\includegraphics[height=4.8cm]{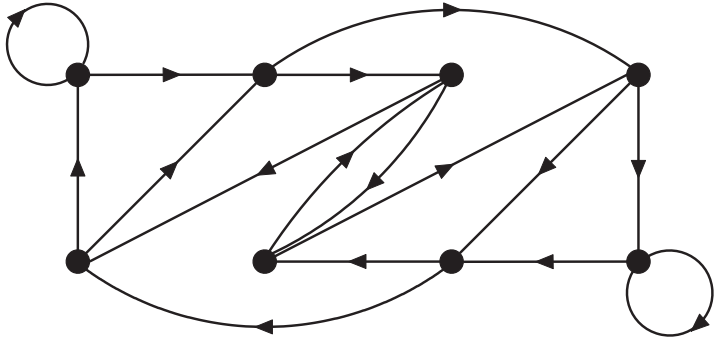}
\caption{A directed graph $H$ such that $Z_H(\cdot)$ is tractable}
\label{fig:figurenew}
\end{figure}
\begin{figure}[h!]
\centering
$$
\AA=\begin{pmatrix}
\hspace{0.08cm}A_{1,1} & & A_{1,3}\\
\hspace{0.08cm}A_{2,1} & & A_{2,3}\\
& & & & A_{3,5} & & A_{3,7}\\
& & & & A_{4,5} & & A_{4,7}\\
& A_{5,2} & & A_{5,4}\\
& A_{6,2} & & A_{6,4}\\
& & & & & A_{7,6} & & A_{7,8}\hspace{0.12cm}\\
& & & & & A_{8,6} & & A_{8,8}\hspace{0.12cm}
\end{pmatrix}
$$
\caption{The $8\times 8$ block-rank-$1$ matrix $\AA$}
\label{fig:figure1}
\end{figure}
\begin{figure}[h!]
\centering
\includegraphics[height=4.2cm]{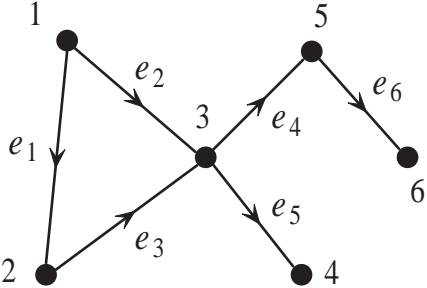}
\caption{The input directed graph $G$} \label{fig:figure2}
\end{figure}\newpage

\section{Proof of Lemma \ref{lemma1}}\label{AppBulatovG}

Bulatov and Grohe showed that for any $m\times m$ non-negative \emph{symmetric}
  matrix $\DD$, $Z_{\DD}(\cdot)$ is \#P-hard if $\DD$ is not block-rank-$1$.
Note that when $\DD$ is symmetric, the directions of the edges in $G$ do not affect
  the value of $Z_\DD(G)$, so we can always assume that $G$ is an undirected graph.

We prove Lemma \ref{lemma1} by giving a reduction from the
  symmetric case.

Let $\CC$ be an $m\times m$ non-negative matrix, which is not block-rank-$1$.
Without loss of generality, we may assume that $\CC_1,\CC_2$,
  the first and the second row vectors of $\CC$, satisfy
  $\CC_1\cdot \CC_2>0$; but $\CC_1$ and $\CC_2$ are not linearly dependent.
Let $\DD$ denote the following symmetric matrix:
$$
D_{i,j}=\CC_i\cdot \CC_j,\ \ \ \ \ \text{for all $i,j\in [m]$}.
$$
By the assumption, we have $D_{1,1},D_{1,2},D_{2,1},D_{2,2}>0$ but
  $D_{1,1}D_{2,2}>D_{1,2}D_{2,1}$.
It then follows from the result of Bulatov and Grohe that $Z_{\DD}(\cdot)$ is \#P-hard to compute.

Now we prove the \#P-hardness of $Z_{\CC}(\cdot)$ by showing a reduction from
  $Z_{\DD}(\cdot)$.
Let $G=(V,E)$ be an input undirected graph of $Z_{\DD}(\cdot)$.
We construct a directed graph $G'=(V',E')$ in which
$$
V'=V\cup \big\{\hspace{0.04cm}w_e:e\in E\hspace{0.04cm}\big\}\ \ \ \ \text{and}\ \ \ \
E'=\big\{\hspace{0.04cm}uw_e,vw_e:e=uv\in E\hspace{0.04cm}\big\}.
$$
By the definition of $Z_{\CC}(\cdot)$ and $Z_{\DD}(\cdot)$, it is easy to verify that\vspace{-0.06cm}
$$
Z_{\CC}(G')=Z_\DD(G),\ \ \ \ \ \text{for any undirected graph $G$.}\vspace{-0.05cm}
$$
As a result, $Z_\DD(\cdot)$ is polynomial-time reducible to $Z_{\CC}(\cdot)$,
  and the latter is also \#P-hard.

\section{Definition of the \genp\ Operation}\label{construction}

In this section, we define the operation \genp.

Let $\calT=\{\hspace{0.01cm}(A_1,B_1),\ldots,(A_r,B_r)\hspace{0.01cm}\}$ be a non-trivial $m\times m$
  block pattern with $r\ge 1$.
We use $\text{diag}(\calT)$ to denote the set of all $i\in [m]$ such that
  $i\in A_k$ and $i\in B_k$ for some $k\in [r]$.
In this section, we always assume that $(\frS,\frT)$ is a $\calT$-pair such that
  every matrix in $\frT$ is block-rank-$1$.
This means that
\begin{enumerate}
\item All matrices in $\frT$ are block-rank-$1$ and are of the same block pattern $\calT$;
\item $\11\in \frS$ and every vector $\ww\in \frS$ is either
\begin{enumerate}
\item[]\ \ \emph{positive}: $w_i>0$ for all $i\in [m]$; or
\item[]\ \ \emph{$\calT$-weakly positive}: $w_i>0$ if and only if $i\in \text{diag}(\calT)$.
\end{enumerate}
\end{enumerate}
Given such a pair $(\frS,\frT)$, \genp\ defines a new $\calP$-pair
$$
(\frS',\frT')=\genp(\frS,\frT),\ \ \ \ \ \text{where $\calP=\gen(\calT)$.}
$$
To this end we first define a pair $(\frS^*,\frT^*)$ from $(\frS,\frT)$,
  which is a \emph{generalized $\calP$-pair} defined as follows.\vspace{0.04cm}

\begin{defi}
Let $\calP$ be an $r\times r$ pattern with $r\ge 1$.
An $r\times r$ nonnegative matrix is called a \emph{$\calP$-diagonal} matrix if
  it is a diagonal matrix and for all $i\in [r]$, its $(i,i)$th entry is
  positive if and only if $(i,i)\in \calP$.

We call $(\frS^*,\frT^*)$ a \emph{generalized $\calP$-pair} if
\begin{enumerate}
\item $\frS^*$ is a nonempty \emph{(}and possibly infinite\emph{)} set of
  \emph{positive} and \emph{$\calP$-weakly positive} vectors with $\11\hspace{-0.06cm}\in\hspace{-0.04cm} \frS^*$;\vspace{-0.08cm}
\item $\frT^*$ is a nonempty \emph{(}and possibly infinite\emph{)} set of
  $\calP$-matrices and \emph{$\calP$-diagonal} matrices.
\end{enumerate}
For any block pattern $\calT$, one can define \emph{$\calT$-diagonal} matrices
  and \emph{generalized $\calT$-pairs} similarly,
  by rep\-lacing the pattern $\calP$ above with the one defined by $\calT$.\vspace{0.04cm}
\end{defi}\newpage

We then use $(\frS^*,\frT^*)$ to define $(\frS',\frT')$.
In this section we only show that $(\frS',\frT')$ is a $\calP$-pair and $\frS'$ is \emph{closed}.
We will give the polynomial-time reduction from $(\frS',\frT')$ to $(\frS,\frT)$ in Appendix \ref{reduction}.

\subsection{Definition of $\frT^*$}

We define $\frT^*$ which contains both $\calP$-matrices and $\calP$-diagonal matrices,
  where $\calP=\gen(\calT)$.\vspace{0.003cm}

There are two types of matrices in $\frT^*$.
First, $\DD$ is an $r\times r$ $\calP$-matrix in $\frT^*$ if there exist
\begin{enumerate}
\item a finite subset of matrices $\{\CC^{[1]},\ldots,
  \CC^{[g]}\}\subseteq \frT$ with $g\ge 1$, and positive integers $s_1,\ldots,s_g$;\vspace{-0.06cm}
\item a finite subset of matrices $\{\DD^{[1]},\ldots,\DD^{[h]}\}\subseteq \frT$ with $h\ge 1$,
  and positive integers $t_1,\ldots,t_h$;\vspace{-0.06cm}
\item a \emph{positive} vector $\ww\in \frS$,
\end{enumerate}
such that: Let $(\balpha^{[i]},\bbeta^{[i]})$ and $(\boldsymbol{\gamma}^{[i]},\bdelta^{[i]})$
  be the representations of $\CC^{[i]}$ and $\DD^{[i]}$, respectively, then\vspace{0.1cm}
$$
D_{i,j}=\sum_{x\in B_i\cap A_j} \Big(\beta^{[1]}_x\Big)^{s_1}\cdots
  \Big(\beta^{[g]}_x\Big)^{s_g}\cdot \Big(\gamma^{[1]}_x\Big)^{t_1}
  \cdots \Big(\gamma^{[h]}_x\Big)^{t_h}\cdot w_x,\ \ \ \ \ \text{for all $i,j\in [r]$.}\vspace{0.05cm}
$$
The following lemma is easy to prove.\vspace{0.05cm}

\begin{lemm}
If $\ww\in \frS$ is positive, then the matrix $\DD$ defined above is a \emph{$\calP$-matrix},
  where $\calP=\text{\emph{\texttt{gen}}}(\calT)$.
\end{lemm}
\begin{proof}
Because $(\frS,\frT)$ is a $\calT$-pair, all the matrices
  $\CC^{[i]}$ and $\DD^{[j]}$, $i\in [g]$ and $j\in [h]$, are $\calT$-matrices and thus,
  $\bbeta^{[i]}$ is positive over $B_1\cup\cdots \cup B_r$ and
  $\bgamma^{[j]}$ is positive over $A_1\cup\cdots\cup A_r$.
Since $\ww$ is positive, it is easy
  to check that $D_{i,j}>0$ if and only if $B_i\cap A_j\ne \emptyset$.
\end{proof}

Second, $\DD$ is an $r\times r$ $\calP$-diagonal matrix in $\frT^*$ if there exist
\begin{enumerate}
\item a finite subset of matrices $\{\CC^{[1]},\ldots,
  \CC^{[g]}\}\subseteq \frT$ with $g\ge 1$, and positive integers $s_1,\ldots,s_g$;\vspace{-0.06cm}
\item a finite subset of matrices $\{\DD^{[1]},\ldots,\DD^{[h]}\}\subseteq \frT$ with $h\ge 1$,
  and positive integers $t_1,\ldots,t_h$;\vspace{-0.06cm}
\item a \emph{$\calT$-weakly positive} vector $\ww\in \frS$,
\end{enumerate}
such that: Let $(\balpha^{[i]},\bbeta^{[i]})$ and $(\boldsymbol{\gamma}^{[i]},\bdelta^{[i]})$ be the representation
  of $\CC^{[i]}$ and $\DD^{[i]}$, respectively, then\vspace{0.1cm}
$$
D_{i,j}=\sum_{x\in B_i\cap A_j} \Big(\beta^{[1]}_x\Big)^{s_1}\cdots
  \Big(\beta^{[g]}_x\Big)^{s_g}\cdot \Big(\gamma^{[1]}_x\Big)^{t_1}
  \cdots \Big(\gamma^{[h]}_x\Big)^{t_h}\cdot w_x,\ \ \ \ \ \text{for all $i,j\in [r]$.}\vspace{0.05cm}
$$
Similarly one can show that\vspace{0.05cm}

\begin{lemm}
If $\ww$ is $\calT$-weakly positive, then the matrix $\DD$
  defined above is \emph{$\calP$-diagonal} where $\calP=\text{\emph{\texttt{gen}}}(\calT)$.
\end{lemm}
\begin{proof}
First, we show that $\DD$ is diagonal.
Let $i\ne j$ be two distinct indices in $[r]$.
If $B_i\cap A_j=\emptyset$, then $D_{i,j}$ is trivially $0$.
Otherwise, for every $k\in B_i\cap A_j$,
  we know that $(k,k)$ is not in the pattern defined by
  $\calT$ because $k\in B_i$, $k\in A_j$ but $i\ne j$.
As a result, we have $w_k=0$
  which implies $D_{i,j}=0$ for all $i\ne j\in [r]$.

Second, if $A_i\cap B_i\ne \emptyset$ then $(k,k)$ is in the pattern
  defined by $\calT$  for every $k\in A_i\cap B_i$. This implies that $w_k>0$.
As a result, we have $D_{i,i}>0$ if and only if $A_i\cap B_i\ne \emptyset$.
\end{proof}

\subsection{Definition of $\frS^*$}

Now we define $\frS^*$.
To this end, we first define $\frS^\#$ which is a set of $r$-dimensional
  positive and $\calP$-weakly positive vectors.
We have $\ww^\#\in \frS^\#$ if and only if one of the following four cases is true:
\begin{enumerate}
\item $\ww^\#=\11$;
\item There exist a finite subset $\{\CC^{[1]},\ldots,\CC^{[g]}\}\subseteq \frT$ with $g\ge 1$,
  positive integers $s_1,\ldots,s_g$ and a vector $\ww \in \frS$
  (positive or $\calT$-weakly positive) such that:
Let $(\balpha^{[i]},\bbeta^{[i]})$ be the representation of $\CC^{[i]}$, then\vspace{0.1cm}
$$
w^\#_i=\sum_{x\in A_i} \Big(\alpha^{[1]}_{x}\Big)^{s_1}\cdots
  \Big(\alpha^{[g]}_{x}\Big)^{s_g}\cdot w_x,\ \ \ \ \
  \text{for all $i\in [r]$}.
$$
It can be checked that $\ww^\#$ is positive if $\ww$ is positive
  and $\ww^\#$ is $\calP$-weakly positive if $\ww$ is $\calT$-weakly positive.\vspace{0.02cm}

\item There exist a finite subset $\{\DD^{[1]},\ldots,\DD^{[h]}\}\subseteq \frT$
  with $h\ge 1$, positive integers $t_1,\ldots,t_g$ and a vector $\ww \in \frS$
  (positive or $\calT$-weakly positive) such that:
Let $( \bgamma^{[i]},\bdelta^{[i]})$ be the representation of $\DD^{[i]}$, then\vspace{0.1cm}
$$
w^\#_i=\sum_{x\in B_i} \Big(\delta^{[1]}_{x}\Big)^{t_1}\cdots
  \Big(\delta^{[h]}_{x}\Big)^{t_h}\cdot w_x,\ \ \ \ \
  \text{for all $i\in [r]$}.
$$
Similarly, it can be checked that $\ww^\#$ is positive if $\ww$ is positive
  and $\ww^\#$ is $\calP$-weakly positive if $\ww$~is $\calT$-weakly positive.\vspace{0.02cm}

\item There exist two finite subsets
$\{\CC^{[1]},\ldots,\CC^{[g]}\}\subseteq \frT$ and $\{\DD^{[1]},\ldots,\DD^{[h]}\}\subseteq \frT$ with $g\ge 1$
  and $h\ge$ $1$, positive integers $s_1,\ldots,s_g,t_1,\ldots,t_h$ and a vector $\ww \in \frS$
  (positive or $\calT$-weakly positive) such that:
Let $(\balpha^{[i]},\bbeta^{[i]})$ and $(\bgamma^{[i]},\bdelta^{[i]})$
  be the representations of $\CC^{[i]}$ and $\DD^{[i]}$, respectively, then\vspace{0.08cm}
$$
w^\#_i=\sum_{x\in B_i\cap A_i} \Big(\beta^{[1]}_{x}\Big)^{s_1}\cdots
  \Big(\beta^{[g]}_{x}\Big)^{s_g}\cdot \Big(\gamma^{[1]}_x\Big)^{t_1}
  \cdots \Big(\gamma^{[h]}_x\Big)^{t_h}\cdot w_x,\ \ \ \ \
  \text{for all $i\in [r]$}.
$$
It can be checked that $\ww^\#$ is always a $\calP$-weakly positive vector.\vspace{0.06cm}
\end{enumerate}
This finishes the definition of $\frS^\#$.\vspace{0.015cm}

Set $\frS^*$ is the \emph{closure} of $\frS^\#$: $\ww\in \frS^*$ if and only if
  there exist a finite subset $\{\ww_1,\ldots,\ww_g\} \subseteq \frS^\#$
  and positive integers $s_1,\ldots,s_g$ such that
$$
\ww=\big(\ww_1\big)^{s_1}\circ \cdots \circ \big(\ww_g\big)^{s_g},
$$
where $\circ$ denotes the Hadamard product.
It immediately implies that $\frS^*$ is closed, and any vector in it
  is either positive or $\calP$-weakly positive.
It is also easy to check that $(\frS^*,\frT^*)$ is a \emph{generalized} $\calP$-pair.

\subsection{Definition of $(\frS',\frT')$}

We use $(\frS^*,\frT^*)$ to define $(\frS',\frT')$ as follows.

First, $\frT'$ contains exactly all the $\calP$-matrices in $\frT^*$.

The definition of $\frS'$ is more complicated.
We have $\ww'\in \frS'$ if and only if \vspace{0.05cm}
\begin{enumerate}
\item $\ww'\in \frS^*$; or
\item There exist\vspace{0.05cm}
\begin{enumerate}
\item a finite subset
  of $\calP$-matrices $\{\CC^{[1]},\ldots,\CC^{[g]}\}\subseteq \frT^*$ with $g\ge 0$ (so
  this set could\\ be empty) and $g$ positive integers $s_1,\ldots,s_g$;\vspace{0.04cm}
\item a finite subset of $\calP$-diagonal matrices $\{\DD^{[1]},\ldots,\DD^{[h]}\}\subseteq \frT^*$
  with $h\ge 1$, and $h$\\ positive integers $t_1,\ldots,t_h$;\vspace{0.04cm}
\item and a vector $\ww\in \frS^*$ (which is either positive or $\calP$-weakly positive),\vspace{0.05cm}
\end{enumerate}
such that $\ww'$ satisfies\vspace{0.1cm}
$$
w'_i=w_i\cdot \Big(C^{[1]}_{i,i}\Big)^{s_1}\cdots \Big(C^{[g]}_{i,i}\Big)^{s_g}\cdot
  \Big(D^{[1]}_{i,i}\Big)^{t_1}\cdots \Big(D^{[h]}_{i,i}\Big)^{t_h},\ \ \ \ \ \text{for any $i\in [r]$.}\vspace{0.02cm}
$$
\end{enumerate}
It can be checked that every $\ww'\in \frS'$ is either positive or $\calP$-weakly positive.\vspace{0.02cm}

This finishes the definition of $(\frS',\frT')$ and the \genp\ operation.
It is easy to verify that the new pair $(\frS',\frT')$ is a $\calP$-pair.
Moreover, since $\frS^*$ is closed, one can show that $\frS'$ is also closed.
This proved the first part of Lemma \ref{reductionlemm}:

\begin{lemm}\label{partlem2}
Let $(\frS,\frT)$ be a $\calT$-pair for some non-trivial
  block pattern $\calT$.
Suppose every matrix in $\frT$ is block-rank-$1$,
then $(\frS',\frT')=\text{\emph{\genp}}(\frS,\frT)$ is a $\calP$-pair,
  where $\calP=\text{\emph{\gen}}(\calT)$, and $\frS'$ is \emph{closed}.
Moreover, the pair $(\frS^*,\frT^*)$ defined from $(\frS,\frT)$
  is a \emph{generalized} $\calP$-pair and $\frS^*$ is also \emph{closed}.
\end{lemm}

\section{Dichotomy: Tractability}\label{tractability}

In this section, we prove Lemma \ref{tractabilitylemma}, the tractability part of the dichotomy theorem.

Let $(\frS_0,\frT_0)=(\frP,\frQ)$ be a finite $\calT_0$-pair, for some
  block pattern $\calT_0$.
Let
$
(\frS_0,\frT_0), \ldots,(\frS_h,\frT_h)
$
be~a sequence of $h+1$ pairs for some $h\ge 0$,
  $m_0>m_1>\ldots>m_h\ge 1$ be $h+1$ positive integers,
  and $\calT_0$, $\calT_1,\ldots,\calT_h$ be $h+1$ block patterns such that\vspace{0.08cm}
\begin{itemize}
\item[\textbf{R:}]
For every $i\in [0:h]$, $\calT_i$ is an $m_i\times m_i$ block pattern;\\
For every $i\in [h]$, $\calT_i=\text{\genb}(\calT_{i-1})$;\\
Either $\calT_h=\emptyset$ is trivial or every set in $\calT_h$ is a singleton; \\
For every $i\in [h]$, $(\frS_i,\frT_i)=\text{\genp}
  (\frS_{i-1},\frT_{i-1})$ is a $\calT_i$-pair; and \\
For every $i\in [0:h]$, all the matrices in $\frT_i$
  are block-rank-$1$.\vspace{0.08cm}
\end{itemize}
We need to show that $Z_{\frP,\frQ}(\cdot)=Z_{\frS_0,\frT_0}(\cdot)$ can be computed in polynomial time.\vspace{0.01cm}

Let $\calG_0=(G_0,\calV_0,\calE_0)$ be an input labeled directed graph
  of $Z_{\frS_0,\frT_0}(\cdot)$.
By definition we have $\calV_0(v)\in$ $\frS_0$ for all vertices $v\in V(G_0)$, and
  $\calE_0(uv)\in \frT_0$ for all edges $uv\in E(G_0)$.
We further assume that the underlying undirected graph of $G_0$ is connected.
(If $G_0$ is not connected, then we only need to compute $Z_{\frS_0,\frT_0}(\cdot)$ for
  each undirected
  connected component of $G_0$ and multiply them to obtain $Z_{\frS_0,\frT_0}(\calG_0)$.)\vspace{0.015cm}

To compute $Z_{\frS_0,\frT_0}(\calG_0)$,
  we will construct~in polynomial-time a sequence of $h+1$ labeled
  directed graphs $\calG_0,\ldots,\calG_h$.
We will show that these graphs have the following two properties:\vspace{0.05cm}
\begin{enumerate}
\item[\textbf{P}$_1$:] For every $\ell\in [0:h]$, $\calG_\ell=(G_\ell,\calV_\ell,\calE_\ell)$
  is a labeled directed graph such that
  $\calV_\ell(v)\in \frS_\ell$ for all $v\in$\\ $V(G_\ell)$;
  $\calE_\ell(uv)\in \frT_\ell$ for all $uv\in E(G_\ell)$; and the underlying undirected
  graph of $G_\ell$ is connected.
\item[\textbf{P}$_2$:] $Z(\calG_0)=Z(\calG_1)=\cdots=Z(\calG_h).$\vspace{0.05cm}
\end{enumerate}
As a result, to compute $Z(\calG_0)$, one only needs to compute $Z(\calG_h)$.
On the other hand, we do know how to compute $Z(\calG_h)$ in polynomial time.
If $\calT_h$ is trivial, then computing $Z(\calG_h)$ is also trivial.
Otherwise, if every set in $\calT_h$ is a singleton, then one can efficiently enumerate all
  possible assignments of $\calG_h$ with non-zero weight (since the underlying undirected
  graph of $G_h$ is \emph{connected}).
This allows us to compute $Z(\calG_0)=Z(\calG_h)$ in polynomial time.

\subsection{Construction of $\calG'$ from $\calG$}

Let $(\frS,\frT)$ be a $\calT$-pair for some $m\times m$ non-trivial block pattern $\calT$
  such that all the matrices in $\frT$ are block-rank-$1$.
Then by Lemma \ref{partlem2},
  $(\frS',\frT')=\genp(\frS,\frT)$ is a $\calP$-pair where $\calP=\gen(\calT)$.

Let $\calG=(G,\calV,\calE)$ be a labeled directed graph such that
  $\calV(v)\in \frS$ for all $v\in V(G)$; $\calE(uv)\in \frT$~for all $uv\in E(G)$;
  and the underlying undirected graph of $G$ is connected.
We further assume that $G$ is not trivial: $V$ is not a singleton (since for this
  special case, $Z(\calG)$ can be computed trivially).
In this section, we show how to construct a new graph $\calG'=(G',\calV',\calE')$
  in polynomial time such that $\calV'(v)\in \frS'$ for all $v\in V(G')$;
  $\calE'(uv)\in \frT'$ for all $uv\in E(G')$; the underlying undirected graph
  of $G'$ is connected; and\vspace{-0.03cm}
\begin{equation}\label{impeq}
Z(\calG)=Z(\calG').
\end{equation}
Then we can repeatedly apply this construction, starting from $\calG_0$,
  to obtain a sequence of $h+1$ labeled directed graphs $\calG_0,\ldots,\calG_h$ that satisfy
  both \textbf{P}$_1$ and \textbf{P}$_2$.
Lemma \ref{tractabilitylemma} then follows.\vspace{0.015cm}

Now we describe the construction of $\calG'$.
Let $G=(V,E)$ and $\calT=\{(A_1,B_1),\ldots,(A_{n}, B_n)\}$ for some $n\ge 1$,
  then $\calP=\gen(\calT)$ is an $n\times n$ pattern.
The construction of $\calG'$ is divided into two steps,
  just like the definition of $(\frS',\frT')=\genp(\frS,\frT)$
  in Appendix \ref{construction}.
In the first step, we construct a labeled graph $\calG^*=(G^*,\calV^*,\calE^*)$
  from $\calG$ such that
\begin{enumerate}
\item $\calV^*(v)\in \frS^*$ for all $v\in V(G^*)$;
  $\calE^*(uv)\in \frT^*$ for all $uv\in E(G^*)$; and the underlying
  undirected\\ graph of $G^*$ is connected, where $(\frS^*,\frT^*)$ denotes the
  \emph{generalized} $\calP$-pair defined in Appendix \ref{construction}.
\item $Z(\calG^*)=Z(\calG)$.
\end{enumerate}
In the second step, we construct $\calG'$ from $\calG^*$ and show that $Z(\calG')=Z(\calG^*)$.

\subsubsection{Construction of $\calG^*$ from $\calG$}

Let $\calG=(G,\calV,\calE)$ and $G=(V,E)$.
We decompose the edge set using the following equivalence relation:\vspace{0.05cm}

\begin{defi}
Let $e,e'$ be two directed edges in $E$.
We say $e \sim e'$ if there exist a sequence of edges $$e=e_0,e_1,\ldots,e_k=e'$$
  in $E$ such that for all $i\in [0:k-1]$,
$e_i$ and $e_{i+1}$ share either the same head or the same tail.\vspace{0.05cm}
\end{defi}

We divide $E$ into equivalence classes $R_1,\ldots,R_f$ using $\sim$:
$$
E=R_1\cup \ldots\cup R_f,\ \ \ \ \text{for some $f\ge 1$.}
$$
Because the underlying undirected graph of $G$ is connected, there is no
  isolated vertex $v$ in $G$ and thus every vertex $v\in V$ appears as an incident
  vertex of some edge in at least one of the equivalence classes.
This equivalence relation is useful because of the following observation.

\begin{obse}\label{obse1}
\emph{For any $i\in [f]$, the subgraph spanned by $R_i$
  is connected if we view it as an undirect\-ed~graph.
There are three types of vertices in it:\vspace{0.05cm}
\begin{enumerate}
\item Type-L: vertices which only have outgoing edges in $R_i$;
\item Type-R: vertices which only have incoming edges in $R_i$; and
\item Type-M: vertices which have both incoming and outgoing edges in $R_i$.\vspace{0.05cm}
\end{enumerate}
Let $\xi:V\rightarrow [m]$ be any assignment with
  $\text{wt}(\calG,\xi)\ne 0$,
then for any $i\in [f]$ there exists a unique $k_i\in [n]$ such that the
  value of every edge $uv\in R_i$ is derived from the $k_i$-th block of $\calT$:
$$
\xi(u)\in A_{k_i}\ \ \ \ \text{and}\ \ \ \ \xi(v)\in B_{k_i}.
$$
\noindent Therefore, for every $i\in [f]$, there exists a unique $k_i\in [n]$ such that\vspace{0.1cm}
\begin{enumerate}
\item For every Type-L vertex $v$ in the graph spanned by $R_i$, $\xi(v)\in A_{k_i}$;\vspace{-0.06cm}
\item For every Type-R vertex $v$ in the graph spanned by $R_i$, $\xi(v)\in B_{k_i}$; and\vspace{-0.06cm}
\item For every Type-M vertex $v$ in the graph spanned by $R_i$, $\xi(v)\in A_{k_i}\cap B_{k_i}$.\vspace{0.1cm}
\end{enumerate}}
\end{obse}

Now we build $\calG^*=(G^*,\calV^*,\calE^*)$, where $G^*=(V^*,E^*)$.
We start with the construction of $G^*$.
$V^*$ is exactly $[f]$ in which the vertex $i\in [f]$ corresponds to $R_i$ of $G$.
For every vertex $v\in V$, if it appears in both the subgraph spanned by $R_i$
  and the one spanned by $R_j$ for some $i\ne j\in [f]$ (note that it cannot appear in more than two
  such subgraphs)
  and if the incoming edges of $v$ are from $R_i$ and the outgoing edges
  of $v$ are from $R_j$, then we add a directed edge $ij$ in $E^*$. Note that
  $E^*$ may have parallel edges.
This finishes the construction of $G^*$.
It is easy to verify that the underlying undirected graph of $G^*$
  is also \emph{connected}. \vspace{0.01cm}

The only thing left is to label the graph $G^*$ with vertex and edge weights.
For every edge in $E^*$ we assign it the following $n\times n$ matrix $\DD$.
Assume the edge $ij$ is created because of $v\in V$, which appears in both
  $R_{i}$ and $R_j$.
Let the incoming edges of $v$ be $u_1v,\ldots,u_sv$ in $R_i$
  and the outgoing edges of $v$ be $vw_1,\ldots,vw_t$ in $R_j$,
  where $s,t\ge 1$.
We use $\CC^{[i]}\in \frT$ to denote the edge weight
  of $u_iv$, $\DD^{[i]}\in \frT$ to denote the edge weight of $vw_i$,
  and $\ww\in \frS$ to denote the vertex weight of $v$ in $\calG$.
We also use $(\balpha^{[i]},\bbeta^{[i]})$ and $(\bgamma^{[i]},\bdelta^{[i]})$ to denote the representations
  of $\CC^{[i]}$ and $\DD^{[i]}$, respectively.
Then the $(i,j)$th entry of $\DD$ is\vspace{0.08cm}
$$
D_{i,j}=\sum_{x\in B_i\cap A_j} \beta^{[1]}_x\cdots \beta^{[s]}_x \cdot \gamma^{[1]}_x\cdots
  \gamma^{[t]}_x\cdot w_x,\ \ \ \ \ \text{for all $i,j\in [n]$.}
$$
By the definition of \genp, it is easy to check that $\DD\in \frT^*$.\vspace{0.015cm}

Finally, we define the vertex weight of $i\in  [f]$. To this end,
  we first define an $n$-dimensional vector $\ww^{[v]}$ for each vertex $v\in V$
  that \emph{only appears} in $R_i$.
We then multiply (using Hadamard product) all such vectors to get the
  vertex weight vector of $i\in [f]$.

Let $v\in V$ be a vertex which only appears in $R_i$, then we have the following three cases:\vspace{0.05cm}
\begin{enumerate}
\item If $v$ is Type-L, then we use $vw_1,\ldots,vw_s$ to denote its outgoing edges.
We let $\ww$ denote the vertex\\ weight of $v$ in $\calG$ and
  $\CC^{[j]}$ denote the edge weight of $vw_j$ with representation $(\balpha^{[j]},\bbeta^{[j]})$.
Then
$$
w^{[v]}_k=\sum_{x\in A_k}\hspace{0.04cm} \alpha^{[1]}_x\cdots \alpha^{[s]}_x\cdot w_x,
\ \ \ \ \ \text{for all $k\in [n]$.}
$$
\item If $v$ is Type-R, then we use $u_1v,\ldots, u_sv$ to denote its
  incoming edges.
We let $\ww$ denote the vertex\\ weight of $v$ in $\calG$ and
  $\CC^{[j]}$ denote the edge weight of $u_jv$ with representation $(\balpha^{[j]},\bbeta^{[j]})$.
Then
$$
w^{[v]}_k=\sum_{x\in B_k}\hspace{0.04cm} \beta^{[1]}_x\cdots \beta^{[s]}_x\cdot w_x,
\ \ \ \ \ \text{for all $k\in [n]$.}
$$
\item If $v$ is Type-M, then we use $u_1v,\ldots,u_sv,vw_1,\ldots,vw_t$ to
  denote its edges where $s,t\ge 1$.
We let $\ww$ be the vertex weight of $v$ in $\calG$, $\CC^{[j]}$ be
  the edge weight of $u_jv$ with representation $(\balpha^{[j]},\bbeta^{[j]})$,
  and $\DD^{[j]}$ be the edge weight of $vw_j$ with representation $(\bgamma^{[j]},\bdelta^{[j]})$.
Then\vspace{0.04cm}
$$
w^{[v]}_k=\sum_{x\in B_k\cap A_k} \beta^{[1]}_x\cdots \beta^{[s]}_x\cdot
  \gamma^{[1]}_x\cdots \gamma^{[t]}_x\cdot w_x,\ \ \ \ \ \text{for all $k\in [n]$.}
$$
\end{enumerate}
We then multiply (using Hadamard product) all the vectors $\ww^{[v]}$ over all vertices
  $v$ that only appear in $R_i$ to get the vertex weight vector $\ww$ of $i\in [f]$ in $\calG^*$.
By definition, it can be checked that $\ww\in \frS^*$.
This finishes the construction of $\calG^*$. Next, we show that $Z(\calG^*)=Z(\calG)$.\vspace{0.02cm}

Let $\phi:V^*=[f]\rightarrow [n]$ be any assignment. We use $\Xi_\phi$ to denote\vspace{0.03cm}
$$
\Big\{\hspace{0.05cm}\xi:V\rightarrow [m]\hspace{0.12cm}\Big|\hspace{0.12cm}
  \forall\hspace{0.04cm}i\in [f],\ \forall\hspace{0.06cm}
  uv\in R_i,\ \xi(u)\in A_{\phi(i)}\ \hspace{0.06cm} \text{and}\ \hspace{0.06cm}\xi(v)\in
  B_{\phi(i)}\hspace{0.02cm}\Big\}.\vspace{0.06cm}
$$
Equivalently, $\phi$ defines for each vertex $v\in V$ a set $U_v\subseteq [m]$, where
\begin{enumerate}
\item If $v$ appears in both the subgraph spanned by $R_i$ and the subgraph spanned by $R_j$,
  for some \\$i\ne j\in [f]$;
  and $v$ is Type-R in $R_i$ and Type-L in $R_j$, then $U_v=B_{\phi(i)}\cap A_{\phi(j)}$;
\item Otherwise, assume $v$ only appears in the subgraph spanned by $R_i$. Then
\begin{enumerate}
\item If $v$ is Type-L, then $U_v=A_{\phi(i)}$;
\item If $v$ is Type-R, then $U_v=B_{\phi(i)}$; and
\item If $v$ is Type-M, then $U_v=B_{\phi(i)} \cap A_{\phi(i)}$,
\end{enumerate}\end{enumerate}
  such that
$
\xi\in \Xi_\phi\hspace{0.05cm} \Longleftrightarrow
  \hspace{0.05cm} \xi(v)\in U_v\ \text{for all $v\in V$.}
$
In particular, $\Xi_\phi=\emptyset$ if $U_v=\emptyset$ for some $v\in V$.

By Observation \ref{obse1}, if $\text{wt}(\calG,\xi)\ne 0$ then $\xi\in \Xi_\phi$ for some unique $\phi$.
For any $v\in V$, we let $\ww^{[v]}$ denote its vertex weight in $\calG$;
  and for any $uv\in E$, we let $\DD^{[uv]}$ denote its edge weight in $\calG$,
  with representation $(\balpha^{[uv]},\bbeta^{[uv]})$.
Then by the definition of $\Xi_\phi$, we have for all $\xi\in \Xi_\phi$,
$$
D^{[uv]}_{\xi(u),\xi(v)}= \alpha^{[uv]}_{\xi(u)} \cdot \beta^{[uv]}_{\xi(v)},
\ \ \ \ \ \text{for all $uv\in E$.}
$$
Therefore, we have the following equation:
$$
\sum_{\xi\in \Xi_\phi}\hspace{0.04cm} \text{wt}(\calG,\xi)
=\sum_{\xi\in \Xi_\phi}\hspace{0.02cm}\left( \prod_{v\in V} w^{[v]}_{\xi(v)}\hspace{0.04cm}
  \prod_{uv\in E} \alpha^{[uv]}_{\xi(u)}\cdot \beta^{[uv]}_{\xi(v)}\right).
$$
This sum can be written as a product:
$$
\sum_{\xi\in \Xi_\phi} \hspace{0.04cm}\text{wt}(\calG,\xi)
  = \prod_{v\in V} H_v,
$$
in which for every $v\in V$, the factor $H_v$ is a sum over $\xi(v)\in U_v$.

By the construction of $\calG^*$, we can show that
\begin{equation}\label{uuuttt}
\text{wt}(\calG^*,\phi)=\sum_{\xi\in \Xi_\phi}\hspace{0.04cm} \text{wt}(\calG,\xi)=\prod_{v\in V} H_v.
\end{equation}
This follows from the following observations:
\begin{enumerate}
\item If $v$ appears in both the subgraph spanned by $R_i$ and the subgraph
  spanned by $R_j$, for some\\ $i\ne j\in [n]$,
  and this $v$ defines an edge $ij\in E^*$, then the edge weight of this edge $ij$ in $\calG^*$ with\\ respect to
  $\phi$ is exactly $H_v$;

\item For every $i\in [n]$, we let $V_i\subseteq V$ denote the set of vertices
  that only appear in the subgraph\\ spanned by $R_i$.
We also let $\ww$ denote the vertex weight of $i\in [n]$ in $\calG^*$.
Then we have
$$
w_{\xi(i)}=\prod_{v\in V_i} H_{v}.
$$
\end{enumerate}\newpage
\noindent As a result, it follows from (\ref{uuuttt}) that\vspace{0.1cm}
$$
Z(\calG^*)=\sum_{\phi}\hspace{0.04cm} \text{wt}(\calG^*,\phi)
  =\sum_{\phi}\hspace{0.04cm} \sum_{\xi\in \Xi_\phi}\hspace{0.04cm} \text{wt}(\calG,\xi)=Z(\calG).\vspace{0.06cm}
$$

\subsubsection{Construction of $\calG'$ from $\calG^*$}

Let $\calG^*=(G^*,\calV^*,\calE^*)$ be the labeled directed graph
  constructed above, where $G^*=(V^*,E^*)$.
We know that $\calV^*(v)\in \frS^*$ for all $v\in V^*$;
  $\calE^*(uv)\in \frT^*$ for all $uv\in E^*$; and the underlying
  undirected graph of $G^*$ is connected.
Since $(\frS^*,\frT^*)$ is a generalized $\calP$-pair, every $\DD\in \frT^*$ is either
  a $\calP$-matrix or a $\calP$-diagonal matrix.\vspace{0.015cm}

We will build a new labeled directed graph
  $\calG'=(G',\calV',\calE')$ with $G'=(V',E')$
  such that $\calV'(v)\in \frS'$ for all $v\in V'$; $\calE'(uv)\in \frT'$ for all $uv \in E'$;
  the underlying undirected graph of $G'$ is connected;
  and $$Z(\calG')=Z(\calG^*).$$

Let $E^*=E_0\cup E_1$, where $E_0$ consists of the edges in $E^*$ whose
  weight is a $\calP$-matrix and $E_1$ consists of the edges in $E^*$ whose weight
  is a $\calP$-diagonal matrix.
We decompose the vertex set $V^*$ of $G^*$ using the following equivalence
  relation $\sim$.\vspace{0.05cm}

\begin{defi}
Let $v,v'$ be two distinct vertices in $V^*$.
$v\sim v'$ if $v$ and $v'$ are connected by $E_1$ \emph{(}which is
  viewed as a set of undirected edges here\emph{)}.\vspace{0.05cm}
\end{defi}

By using $\sim$, we divide $V^*$ into equivalence
  classes $V_1,\ldots,V_g$ for some $g\ge 1$.
This relation is useful because of the following observation:\vspace{0.04cm}

\begin{obse}\label{obse2}
\emph{Let $\phi:V^*\rightarrow [n]$ be an assignment with non-zero weight:\hspace{-0.05cm}
  $\text{{wt}}(\calG^*,\phi)\ne 0$.
Then for any $i\in [g]$, there exists a unique $k_i\in [n]$ such that
  $\phi(v)=k_i$ for all $v\in V_i$.
}\end{obse}

Now we construct $\calG'=(G',\calV',\calE')$.
First we construct $G'=(V',E')$.
$V'$ is exactly $[g]$ in which vertex $i\in [g]$ corresponds to $V_i$.
For every edge $uv\in E_0$ such that $u\in V_i$, $v\in V_j$,
  and $i\ne j\in [g]$, we add an edge from $i$ to $j$ in $G'$.
This finishes the construction of $G'$.
It is easy to verify that the underlying undirected graph
  of $G'$ is also \emph{connected}.\vspace{0.01cm}

Finally, we assign vertex and edge weights.
For each edge $ij$ in $G'$, suppose it is created because of $uv\in E_0$.
Then the edge weight of $ij$ is the same as that of $uv$.
As a result, all the edge weight matrices of $\calG'$ come from $\frT'$ (since
  by definition of \genp, $\frT'$ contains all the $\calP$-matrices in $\frT^*$).\vspace{0.02cm}

We define the vertex weights of $\calG'$ as follows.
If $V_i=\{v\}$ is a singleton, then the vertex weight of $i$ in $\calG'$
  is the same as the weight of $v$ in $\calG^*$.
Otherwise, we let $v_1,\ldots,v_r$ be the vertices in $V_i$ with $r> 1$,
  let $e_1,\ldots,e_s$ be the edges in $E_1$ with both vertices in $V_i$ for some $s\ge 1$,
  and let $e_1',\ldots,e_t'$ be the edges in $E_0$ with both vertices in $V_i$
  for some $t\ge 0$.
We use $\ww^{[j]}\in \frS^*$ to denote the vertex weight of $v_j$ in $\calG'$
  $\CC^{[j]}\in \frT^*$ to denote the $\calP$-diagonal matrix of $e_j$
  and $\DD^{[j]}\in \frT^*$ to denote the $\calP$-matrix of $e_j'$.
Then we assign the following vertex weight vector $\ww$ to $i\in V'$:\vspace{0.06cm}
$$
w_k=w^{[1]}_k\cdots w^{[r]}_k\cdot C^{[1]}_{k,k}\cdots C^{[s]}_{k,k}\cdot
  D^{[1]}_{k,k}\cdots D^{[t]}_{k,k},\ \ \ \ \ \text{for every $k\in [n]$.}\vspace{0.06cm}
$$
By definition, we have $\ww\in \frT'$.
Using Observation \ref{obse2}, it is also easy to verify that $Z(\calG')=Z(\calG^*)$.\vspace{0.005cm}

This completes the proof of Lemma \ref{tractabilitylemma}.

\section{Reduction: Normalized Matrices are Free to Use}

To give a polynomial-time reduction from $(\frS',\frT')=\genp(\frS,\frT)$
  to $(\frS,\frT)$, we need to first prove a technical lemma on
  \emph{normalized} block-rank-$1$ matrices.\vspace{0.01cm}

Let $\CC$ be an $m\times m$ block-rank-$1$ matrix of block pattern $\calT$ and
  representation $(\balpha,\bbeta)$, where $\calT=\{(A_1, B_1),\ldots,(A_r,B_r)\}$
  for some $r\ge 1$.
By definition, $\balpha$ satisfies
$$
\sum_{j\in A_i} \alpha_j=1,\ \ \ \ \ \text{for all $i\in [r]$.}
$$
We say $\CC'$ is the \emph{normalized} version of $\CC$ if it is an $m\times m$
  block-rank-$1$ matrix of block pattern $\calT$ and representation $(\balpha,\bdelta)$, where
$$
\delta_j=\frac{\beta_j}{\sum_{k\in B_i} \beta_k},\ \ \ \ \ \text{for all $j\in B_i$ and
  $i\in [r]$,}
$$
so that $\bdelta$ also satisfies
$$
\sum_{j\in B_i} \delta_j=1,\ \ \ \ \ \text{for all $i\in [r]$.}\vspace{0.04cm}
$$

Let $(\frP,\frQ)$ be a finite $\calT$-pair for some non-trivial
  $m\times m$ block pattern $\calT$,
  and
$$
\frQ=\big\{\CC^{[1]},\ldots,\CC^{[s]}\big\},
$$
in which every $\CC^{[i]}$ is block-rank-$1$ and has
  representation $(\balpha^{[i]},\bbeta^{[i]})$.
For each $i\in [s]$, we let $\DD^{[i]}$ denote the normalized version
  of $\CC^{[i]}$ with representation $(\balpha^{[i]},\bdelta^{[i]})$, and
$$
\frQ'=\big\{\CC^{[1]},\ldots,\CC^{[s]},\DD^{[1]},\ldots,\DD^{[s]}\big\}.
$$
In this section, we prove the following technical lemma:\vspace{0.06cm}

\begin{lemm}\label{techlemma}
$Z_{\frP,\frQ}(\cdot)$ and $Z_{\frP,\frQ'}(\cdot)$ are computationally equivalent.\vspace{0.06cm}
\end{lemm}
\begin{proof}
In the proof, we use two levels of interpolations and Vandermonde systems.\vspace{0.01cm}

We start with some notation.
Let $\calG=(G,\calV,\calE)$ be the input labeled directed graph of $Z_{\frP,\frQ'}(\cdot)$ with $G=(V,E)$.
For $v\in V$, we use $\ww^{[v]}\in \frP$ to denote its vertex weight.
We use $E_i\subseteq E$, $i\in [s]$, to denote the set of edges labeled with $\CC^{[i]}$,
  and $F_i\subseteq E$, $i\in [s]$, to denote the set of edges labeled with $\DD^{[i]}$.
For every assignment $\xi:V\rightarrow [m]$, we define\vspace{0.09cm}
$$
\text{vw}(\xi)=\prod_{v\in V}\hspace{0.05cm} w^{[v]}_{\xi(v)},\ \ \ \ \
\text{cw}(\xi)=\prod_{i\in [s]}\hspace{0.08cm}\prod_{uv\in E_i} C^{[i]}_{\xi(u),\xi(v)},\ \ \ \ \
\text{dw}(\xi)=\prod_{i\in [s]}\hspace{0.08cm}\prod_{uv\in F_i} D^{[i]}_{\xi(u),\xi(v)}.\vspace{0.04cm}
$$
Note that a product over an empty set is equal to $1$.

\noindent Then we need to compute the following sum
$$
Z_{\frP,\frQ'}(\calG)=\sum_{\xi}\hspace{0.06cm} \text{vw}(\xi)
  \cdot \text{cw}(\xi)\cdot \text{dw}(\xi).
$$

For all $a\in [s]$ and $b\in [r]$, we use $K^{[a]}_b>0$ to denote the number
  such that
$$
C^{[a]}_{i,j}=K^{[a]}_b\cdot D^{[a]}_{i,j},\ \ \ \ \ \text{for all $i\in A_b$ and $j\in B_b$.}
$$
Actually, this gives us the following equation
$$
C^{[a]}_{i,j}=K^{[a]}_b\cdot D^{[a]}_{i,j},\ \ \ \ \ \text{for all $i\in A_b$ and
  $j\in [m]$,}\vspace{0.05cm}
$$
since $\CC^{[a]}$ and $\DD^{[a]}$ have the same block pattern $\calT$.
Then we use $\text{kw}(\xi)$, where $\xi:V\rightarrow [m]$, to denote\vspace{0.1cm}
$$
\text{kw}(\xi)=\prod_{a\in [s]}\hspace{0.1cm}\left( \prod_{uv\in F_a\hspace{0.06cm} \text{with}
  \hspace{0.09cm} \xi(u)\in A_b} K^{[a]}_b\hspace{-0.02cm}\right).\vspace{0.05cm}
$$
We use $X$ to denote the set of all possible values of $\text{kw}(\xi)$:
$$
X=\big\{\hspace{0.05cm}\text{kw}(\xi)\hspace{0.1cm}\big|\hspace{0.1cm}
  \xi:V\rightarrow [m]\hspace{0.02cm}\big\}.
$$
It can be checked that $|X|$ is polynomial in $|E|$ since both $s$ and $r$
  are considered as constants here.
We use $L$ to denote $|X|$.\vspace{0.02cm}

For all $k\in [0:L-1]$, we build a new graph $\calG^{[k]}=(G^{[k]},\calV^{[k]},\calE^{[k]})$,
  where $G^{[k]}=(V^{[k]},E^{[k]})$:\vspace{0.09cm}
\begin{enumerate}
\item $V\subseteq V^{[k]}$ and every $v\in V$ is labeled with the same vertex weight
  as in $\calG$;
\item For all $i\in [s]$ and $uv\in E_i$, we add one edge $uv\in E^{[k]}$ and
  label it with the same matrix $\CC^{[i]}$;
\item For all $i\in [s]$ and all $e=uv\in F_i$, we add $L-k$ parallel edges from
  $u$ to $v$ with $\CC^{[i]}$ as their edge weights;
  we also add $2k$ new vertices $u_{e,j}$ and $v_{e,j}$, $j\in [k]$, to $V^{[k]}$; we add one edge from
    $u$ to $u_{e,j}$ and one edge from $v_{e,j}$ to $v$ for all $j\in [k]$, all of which are labeled
    with $\CC^{[i]}$.
For each new vertex, we assign $\mathbf{1}$ as its vertex weight.\vspace{0.09cm}
\end{enumerate}
It is clear that $\calG^{[k]}$ can be constructed in polynomial time
  and is a valid input of $Z_{\frP,\frQ}(\cdot)$.\vspace{0.005cm}

Fix $k\in [0:L-1]$.
For every assignment $\phi:V\rightarrow [m]$, we let $\Xi_\phi$
  denote the set of all $\xi:V^{[k]}\rightarrow [m]$ such that
  $\xi(v)=\phi(v)$ for all $v\in V$.
We also define
$$
\text{wt}^{[k]}(\phi)=\sum_{\xi\in \Xi_\phi }\hspace{0.05cm} \text{wt}(\calG^{[k]},\xi).
$$ \noindent
Then we have the following equation\vspace{0.05cm}
$$
Z_{\frP,\frQ}(\calG^{[k]})=\sum_{\xi:V^{[k]}\rightarrow [m]}\hspace{-0.05cm} \text{wt}(\calG^{[k]},\xi)
  =\sum_{\phi:V\rightarrow [m]} \text{wt}^{[k]}(\phi).\vspace{0.1cm}
$$

By the construction, we show that
\begin{equation}\label{hahaeq}
\text{wt}^{[k]}(\phi)=\text{vw}(\phi)\cdot \text{cw}(\phi)\cdot \Big(\text{dw}(\phi)\Big)^L
  \cdot \Big(\text{kw}(\phi)\Big)^{L+k},\ \ \ \ \ \text{for all $k\in [0:L-1]$}.
\end{equation}
First, we have
\begin{equation}\label{hahaeq2}
\text{wt}^{[k]}(\phi) = \text{vw}(\phi) \cdot \text{cw}(\phi) \cdot
\sum_{\xi\in \Xi_\phi} \left( \prod_{i\in [s]}\left( \prod_{e=uv\in F_i}
  \left(C^{[i]}_{\xi(u),\xi(v)}\right)^{L-k} \left(\prod_{j\in [k]} C^{[i]}_{\xi(u),\xi(u_{e,j})}
  C^{[i]}_{\xi(v_{e,j}),\xi(v)}\right) \right) \right).\vspace{0.12cm}
\end{equation}
For each edge $e=uv\in F_i$ for some $i\in [s]$, there must exist
  an index $b_e\in [r]$ such that $\phi(u)\in A_{b_e}$ and $\phi(v)\in B_{b_e}$;
  otherwise both sides of (\ref{hahaeq}) are $0$ and we are done.
In this case, the sum in (\ref{hahaeq2}) becomes\vspace{0.06cm}
\begin{equation}\label{hahaeq3}
\prod_{i\in [s]} \left(\prod_{e=uv\in F_i} \left(K^{[i]}_{b_e}\cdot D^{[i]}_{\xi(u),\xi(v)}\right)^{L-k}
\left(\sum_{x\in B_{b_e}} C^{[i]}_{\xi(u),x} \right)^k \left(
\sum_{x\in A_{b_e}} C^{[i]}_{x,\xi(v)}\right)^k\right).\vspace{0.06cm}
\end{equation}
By the definition of $(\balpha^{[i]},\bbeta^{[i]})$ and $(\balpha^{[i]},\bdelta^{[i]})$, we have\vspace{0.06cm}
$$
\sum_{x\in B_{b_e}} C^{[i]}_{\xi(u),x}=\alpha^{[i]}_{\xi(u)}\sum_{x\in B_{b_e}} \beta^{[i]}_{x}
  = \alpha^{[i]}_{\xi(u)}\cdot K^{[i]}_{b_e}\ \ \ \ \ \text{and}\ \ \ \ \
\sum_{x\in A_{b_e}} C^{[i]}_{x,\xi(v)} = \beta^{[i]}_{\xi(v)}.\vspace{0.02cm}
$$
As a result, (\ref{hahaeq3}) becomes\vspace{0.06cm}
$$
\prod_{i\in [s]} \left(\prod_{e=uv\in F_i} \left(K^{[i]}_{b_e}\cdot D^{[i]}_{\xi(u),\xi(v)}\right)^{L-k}
\left(\alpha^{[i]}_{\xi(u)}\cdot K^{[i]}_{b_e}\right)^k \left(
\beta^{[i]}_{\xi(v)}\right)^k\right)=
\prod_{i\in [s]} \left(\prod_{e=uv\in F_i} \left( K^{[i]}_{b_e}\right)^{L+k}
  \left(D^{[i]}_{\xi(u),\xi(v)}\right)^{L}\right).\vspace{0.06cm}
$$
This finishes the proof of equation (\ref{hahaeq}).\vspace{0.01cm}

Since $L$ is polynomial in the input size,
  we can use $Z_{\frP,\frQ}(\cdot)$ as an oracle to compute
$$
\sum_{\phi:V\rightarrow [m]} \text{vw}(\phi)\cdot \text{cw}(\phi)\cdot \Big(\text{dw}(\phi)\Big)^L
  \cdot \Big(\text{kw}(\phi)\Big)^{L+k},\ \ \ \ \ \text{for all $k\in [0:L-1]$}.
$$
in a polynomial number of steps.

For every $x\in X$, we use $\Phi_x$ to denote the set of $\phi:V\rightarrow [m]$ with
  $\text{kw}(\phi)=x$, then we computed\vspace{0.04cm}
$$
\sum_{x\in X} \left(\sum_{\phi\in \Phi_x} \text{vw}(\phi)\cdot \text{cw}(\phi)\cdot
  \Big(\text{dw}(\phi)\Big)^L \right)\cdot x^{L+k},\ \ \ \ \ \text{for all $k\in [0:L-1]$.}\vspace{0.04cm}
$$
Because $x>0$ for all $x\in X$, we can solve this Vandermonde system and obtain
$$
\sum_{\phi\in \Phi_x} \text{vw}(\phi)\cdot \text{cw}(\phi)\cdot \Big(\text{dw}(\phi)\Big)^L,
  \ \ \ \ \ \text{for each $x\in X$,}
$$
in a polynomial number of steps.\vspace{0.008cm}

It is also clear that the whole process can be repeated for any $L'\ge L$ with
$$
L'\le L+\text{poly}(\text{input size}),
$$
and we can use $Z_{\frP,\frQ}(\cdot)$ as an oracle to compute
$$
\sum_{\phi\in \Phi_x} \text{vw}(\phi)\cdot \text{cw}(\phi)\cdot \Big(\text{dw}(\phi)\Big)^{L'},
  \ \ \ \ \text{for all $x\in X$ and $L\le L'\le L+\text{poly}(\text{input size})$,}
$$
in a polynomial number of steps.\vspace{0.005cm}

Next we use $Y$ to denote the set of all possible values of $\text{dw}(\phi)$,
  $\phi:V\rightarrow [m]$ (note it is possible that $0\in Y$).
Again, $|Y|$ is polynomial and we use $M$ to denote $|Y|$.
For every $x\in X$, we can compute
$$
\sum_{\phi\in \Phi_x} \text{vw}(\phi)\cdot \text{cw}(\phi)\cdot \Big(\text{dw}(\phi)\Big)^{L+k},
  \ \ \ \ \text{for all $k\in [0:M-1]$.}
$$
Let $\Phi_{x,y}$ denote the set of $\phi$ with $\text{kw}(\phi)=x$ and $\text{dw}(\phi)=y$.
Solving this Vandermonde system, we get\vspace{0.04cm}
$$
\sum_{\phi\in \Phi_{x,y}} \text{vw}(\phi)\cdot \text{cw}(\phi),
  \ \ \ \ \ \text{for all $x\in X$ and $0<y\in Y$.}
$$
Finally, using all these items, we can compute $Z_{\frP,\frQ'}(\calG)$ in a polynomial number of steps:\vspace{0.06cm}
$$
Z_{\frP,\frQ'}(\calG)=\sum_{x\in X,\hspace{0.06cm}0<y\in Y}
  \left(\sum_{\phi \in \Phi_{x,y}} \text{vw}(\phi)\cdot \text{cw}(\phi)\right)\cdot y.\vspace{0.1cm}
$$
This proves the lemma since the other direction from
  $Z_{\frP,\frQ}(\cdot)$ to $Z_{\frP,\frQ'}(\cdot)$ is trivial.
\end{proof}

\section{Polynomial-Time Reduction from $(\frS',\frT')$ to $(\frS,\frT)$}\label{reduction}

Let $(\frS,\frT)$ be a $\calT$-pair, where $\calT$ is a non-trivial $m\times m$ block pattern
  $\calT=\{(A_1,B_1),\ldots,(A_r,B_r)\}$ with $r\ge 1$ and every matrix in $\frT$ is block-rank-$1$.
Let $\calP$ be the $r\times r$ pattern where $\calP=\gen(\calT)$
  and $(\frS',\frT')$ be the $\calP$-pair generated from $(\frS,\frT)$ using the
  \genp\ operation:
$
(\frS',\frT')=\genp(\frS,\frT).
$
We also use $(\frS^*,\frT^*)$ to denote the
  generalized $\calP$-pair defined in Appendix \ref{construction}.\vspace{0.009cm}

In this section, we prove that $(\frS',\frT')$ is polynomial-time reducible to $(\frS,\frT)$.
To this end, we first reduce $(\frS',\frT')$ to $(\frS^*,\frT^*)$, and then reduce $(\frS^*,\frT^*)$
  to $(\frS,\frT)$.
The first step is trivial, so we will only give a polynomial-time reduction from $(\frS^*,\frT^*)$ to
  $(\frS,\frT)$ below.\vspace{0.01cm}

Let $\frP^*=\{\pp^{[i]}:i\in [s]\}$ be a finite subset of vectors in $\frS^*$ with
  $\11\in \frP^*$ and $\frQ^*=\{\FF^{[i]}:i\in [t]\}$ be a finite subset of matrices in $\frT^*$.
By the definition of \genp, they can be generated by a finite subset
  $\frP=\{\ww^{[i]}:i\in [h]\}\subseteq \frS$ with $\11\in \frP$ and a finite subset
  $\frQ=\{\CC^{[i]}:i\in [g]\}\subseteq \frT$ in the following sense.
(We let $(\balpha^{[i]},\bbeta^{[i]})$ denote the representation of
  $\CC^{[i]}$ for every $i\in [g]$.)\vspace{0.008cm}

For every matrix $\FF\in \frQ^*$,
  there exists a $(2g+1)$-tuple
$$\Big(k \in [h];\kk=(k_1,\ldots,k_g);\ll=(\ell_1,\ldots,\ell_g)\Big),$$
  where $k_i,\ell_i\ge 0$, $\kk\ne \00$ and $\ll\ne \00$, such that\vspace{0.04cm}
\begin{equation}\label{lateruse}
F_{i,j}=\sum_{x\in B_i\cap A_j} \Big(\beta_x^{[1]}\Big)^{k_1}\cdots
  \Big(\beta_x^{[g]}\Big)^{k_g} \cdot \Big(\alpha_x^{[1]}\Big)^{\ell_1}\cdots
  \Big(\alpha_x^{[g]}\Big)^{\ell_g}\cdot w^{[k]}_x.
\end{equation}
This $(2g+1)$-tuple is also call the
  (not necessarily unique) representation of $\FF$ with respect to $(\frP,\frQ)$.\vspace{0.009cm}

For every $\pp\in \frP^*$, there
  exist three finite (and possibly empty) sets
  $\calS_1$, $\calS_2$ and $\calS_3$ of tuples,
  where every tuple in $\calS_1$ and $\calS_2$ is of the form
$$
\Big(k\in [h];\kk=(k_1,\ldots,k_g)\Big)
$$
with $k_i\ge 0$ and $\kk\ne \00$, and every tuple in $\calS_3$ is of the form
$$\Big(k \in [h];\kk=(k_1,\ldots,k_g);\ll=(\ell_1,\ldots,\ell_g)\Big)$$
with $k_i,\ell_i\ge 0$, $\kk\ne \00$ and $\ll\ne \00$.
Every tuple in $\calS_1$ gives us a vector whose $i$th entry, $i\in [r]$, is equal to
$$
\sum_{x\in A_i} \Big(\alpha_x^{[1]}\Big)^{k_1}\cdots
  \Big(\alpha_x^{[g]}\Big)^{k_g} \cdot w^{[k]}_x\hspace{0.04cm};
$$\newpage
\noindent every tuple in $\calS_2$ gives us a vector whose $i$th entry, $i\in [r]$, is equal to
$$
\sum_{x\in B_i} \Big(\beta_x^{[1]}\Big)^{k_1}\cdots
  \Big(\beta_x^{[g]}\Big)^{k_g} \cdot w^{[k]}_x;
$$
and every $(2g+1)$-tuple in $\calS_3$ gives us a vector whose $i$th entry, $i\in [r]$, is equal to
$$
\sum_{x\in B_i\cap A_i} \Big(\beta_x^{[1]}\Big)^{k_1}\cdots \Big(\beta_x^{[g]}\Big)^{k_g}
  \cdot \Big(\alpha_x^{[1]}\Big)^{\ell_1}\cdots \Big(\alpha_x^{[g]}\Big)^{\ell_g}\cdot w^{[k]}_x.
$$
Vector $\pp$ is then the Hadamard product of all these vectors.\vspace{0.005cm}

We remark that all the exponents $k_i,\ell_i$ in the equations above are considered
  as constants, because both $(\frP,\frQ)$ and $(\frP^*,\frQ^*)$ are fixed.
We now prove the following lemma.\vspace{0.05cm}

\begin{lemm}
$Z_{\frP^*,\frQ^*}(\cdot)$ is polynomial-time reducible to $Z_{\frP,\frQ}(\cdot)$.\vspace{0.04cm}
\end{lemm}

\def\frR{\mathfrak{R}}

\subsection{Proof Sketch}

We first give a proof sketch.
Again, we will use interpolations and Vandermonde systems.\vspace{0.003cm}

First, by Lemma \ref{techlemma}, we only need to give a reduction from $Z_{\frP^*,\frQ^*}(\cdot)$
  to $Z_{\frP,{\frR}}(\cdot)$, where
$$
\frR=\Big\{\CC^{[i]},\DD^{[i]}:i\in [g]\hspace{0.01cm}
  \Big\}
$$
contains both $\CC^{[i]}$ and its \emph{normalized} version $\DD^{[i]}$, $i\in [g]$.\vspace{0.006cm}

Let $\calG=(G,\calV,\calE)$ be an input labeled graph of $Z_{\frP^*,\frQ^*}(\cdot)$,
  where $G=(V,E)$.
For every assignment $\xi:V\rightarrow [r]$, we will define
  $\text{nvw}(\xi)>0$.
Moreover, let $X$ be the set of all possible values of $\text{nvw}(\xi)$, and $L=|X|$,
  then $L$ is polynomially bounded.
For every $k\in [L]$, we will build a new labeled directed graph $\calG^{[k]}$ from $\calG$.
$\calG^{[k]}$ is a valid input graph of $Z_{\frP,\frR}(\cdot)$ (with domain $[m]$) and satisfies\vspace{0.05cm}
\begin{equation}\label{lastbig}
Z_{\frP,\frR}(\calG^{[k]})=
  \sum_{\xi:V \rightarrow [r]} \text{wt}(\calG,\xi) \cdot \Big(\text{nvw}
  (\xi)\Big)^k.\vspace{0.06cm}
\end{equation}
\noindent
For each $x\in X$, we use $\Xi_x$ to denote the set of all $\xi:V\rightarrow [r]$
  with $\text{nvw}(\xi)=x$.
Then by solving the Vandermonde system which consists of equations (\ref{lastbig}) for $k=1,2,\ldots,L$,
  we can compute
$$
\sum_{\xi\in \Xi_x}\hspace{0.02cm} \text{wt}(\calG,\xi),\ \ \ \ \ \text{for every $x\in X$,}
$$ \noindent
which allow us to compute in polynomial time
$$
Z_{\frP^*,\frQ^*}(\calG)=\sum_{\xi:V \rightarrow [r]} \text{wt}(\calG,\xi)
  =\sum_{x\in X} \left(\sum_{\xi\in \Xi_x}\hspace{0.02cm}\text{wt}(\calG,\xi)\right).\vspace{-0.05cm}
$$

\subsection{Construction of $\calG^{[k]}$}

We start with the construction of $\calG^{[1]}=(G^{[1]},\calV^{[1]},\calE^{[1]})$.
It will become clear that the construction can be generalized
  to get $\calG^{[k]}$ for every $k\in [L]$.\vspace{0.006cm}

Let $V=[n]$, then
  the vertex set $V^{[1]}$ of $G^{[1]}=(V^{[1]},E^{[1]})$ will be defined as
  a union:
$$
V^{[1]}=R_1\cup R_2\cup \cdots \cup R_n,
$$
where $R_k$ corresponds to vertex $k\in V$ and any edge $uv\in E^{[1]}$
  will be between two vertices $u,v\in V^{[1]}$ such that $u,v\in R_k$ for some unique $k\in [n]$.
$R_i$ and $R_j$, $i\ne j\in [n]$, are not necessarily disjoint and there
  could be vertices shared by (at most) two different sets $R_i$ and $R_j$.
We further divide the vertices of $R_i$, $i\in [n]$, into three types:
  In the subgraph of $G^{[1]}$ spanned by $R_i$,\vspace{0.06cm}
\begin{enumerate}
\item The Type-L vertices only have outgoing edges;\vspace{-0.06cm}
\item The Type-R vertices only have incoming edges; and\vspace{-0.06cm}
\item The Type-M vertices have both incoming and outgoing edges.\vspace{0.06cm}
\end{enumerate}
When adding a new vertex, we will also specify which type it is.
The construction also guarantees that the underlying undirected graph spanned by every $R_i$ is connected.

\subsubsection{Construction of \hspace{0.04cm}$G^{[1]}=(V^{[1]},E^{[1]})$}

We start with the vertex set $V^{[1]}$.\vspace{0.05cm}
\begin{enumerate}
\item First, for every $i\in [n]$ and $a\in [g]$, we add a new
  Type-L vertex $u_{i,a}$ in $R_i$ and add a new Type-R vertex
  $w_{i,a}$ in $R_i$.
All these vertices appear in $R_i$ only.
\item Second, for every $e=ij\in E$, where $i,j\in [n]$, we add a vertex $v_e\in R_i\cap R_j$,
  which is a Type-R vertex in $R_i$ and a Type-L vertex in $R_j$.
\item Finally, for every $i\in V$ let $\pp\in \frP^*$ be its vertex weight in $\calG$.
Then by the discussion earlier, it can be generated from $(\frP,\frQ)$
  using three finite sets of tuples $\calS_1,\calS_2$ and $\calS_3$.
For each tuple $\ss$ in $\calS_1$ we add a new Type-L vertex $v_{i,\ss}$ in $R_i$;
  for each tuple $\ss$ in $\calS_2$, we add a new Type-R vertex in $R_i$;
  and for each tuple $\ss$ in $\calS_3$ we add a new Type-M vertex in $R_i$.
All these vertices appear in $R_i$ only.\vspace{0.05cm}
\end{enumerate}
We will add some more vertices later. Now we start to create
  edges, and assign edge/vertex weights.

First, for every $i\in [n]$, we add $2g$ edges to connect $u_{i,a}$ and
  $w_{i,a}$, $a\in [g]$:\vspace{0.06cm}
\begin{enumerate}
\item For every $a\in [g]$, add one edge from $u_{i,a}$ to $w_{i,a}$,
and label the edge with $\CC^{[1]}$;\vspace{-0.06cm}
\item For every $a\in [g]$, add one edge from $u_{i,a}$ to $w_{i,a+1}$ (with
  $w_{i,g+1}=w_{i,1}$), and label it with $\CC^{[1]}$;\vspace{-0.06cm}
\item For every $a\in [g]$, the vertex weight vector of both $u_{i,a}$
  and $w_{i,a}$ is the all-one vector $\mathbf{1}$.\vspace{0.06cm}
\end{enumerate}

Second, for each edge $e=ij\in E$, we add the incident edges of $v_e\in R_i\cap R_j$
  as follows.
Assume~the edge weight matrix of
  $ij$ in $\calG$ is generated by $(\frP,\frQ)$ using the following $(2g+1)$-tuple:
$$\big(k\in [h];\kk=(k_1,\ldots,k_g);\ll=(\ell_1,\ldots,\ell_g)\big),$$
  where $k_i,\ell_i\ge 0$, $\kk\ne \00$ and $\ll\ne \00$.
Then we add the following incident edges of $v_e$: \vspace{0.06cm}
\begin{enumerate}
\item For each $b\in [g]$, we add $k_b$ parallel edges from
  $u_{i,b}$ to $v_e$ in $R_i$, all of which are labeled with $\CC^{[b]}$;\vspace{-0.06cm}
\item For each $b\in [g]$, we add $\ell_b$ parallel edges from
  $v_e$ to $w_{j,b}$ in $R_j$, all of which are labeled with $\CC^{[b]}$;\vspace{-0.06cm}
\item Assign the vertex weight vector $\ww^{[k]}\in \frP$ to $v_e$.\vspace{0.06cm}
\end{enumerate}

Finally, for every vertex $i\in V$ we use $\pp$ to denote its vertex weight in $\calG$.
Assume $\pp$ is generated by $(\frP,\frQ)$ using three finite sets $\calS_1,\calS_2$ and $\calS_3$
  of tuples.
For each $\ss=(k\in [h];\kk=(k_1,\ldots,k_g))$ in $\calS_1$ with $k_i\ge 0$
  and $\kk\ne \00$, we already added a Type-L vertex $v_{i,\ss}$ in $R_i$
  (which appears in $R_i$ only).
We add the following incident edges of $v_{i,\ss}$:\vspace{0.06cm}
\begin{enumerate}
\item For each $b\in [g]$, add $k_b$ parallel edges from
  $v_{i,\ss}$ to $w_{i,b}$ in $R_i$, all of which are labeled with $\CC^{[b]}$;\vspace{-0.06cm}
\item Assign the vertex weight vector $\ww^{[k]}\in \frP$ to $v_{i,\ss}$.\vspace{0.06cm}
\end{enumerate}
For every $\ss=(k\in [h];\kk=(k_1,\ldots,k_g))$
  in $\calS_2$, we already added a Type-R vertex $v_{i,\ss}\in R_i$.
We add the following incident edges of $v_{i,\ss}$ in $R_i$:\vspace{0.06cm}
\begin{enumerate}
\item For each $b\in [g]$, add $k_b$ parallel edges from
  $u_{i,b}$ to $v_{i,\ss}$ in $R_i$, all of which are labeled with $\CC^{[b]}$;\vspace{-0.06cm}
\item Assign the vertex weight vector $\ww^{[k]}\in \frP$ to $v_{i,\ss}$.\vspace{0.06cm}
\end{enumerate}
For every tuple $\ss=(k\in [h];\kk=(k_1,\ldots,k_g);\ll=(\ell_1,\ldots,\ell_g))$
  in $\calS_3$, we already added a Type-M vertex $v_{i,\ss}$ in $R_i$.
We add the following incident edges of $v_{i,\ss}$ in $R_i$:\vspace{0.06cm}
\begin{enumerate}
\item For every $b\in [g]$, add $k_b$ parallel edges from
  $u_{i,b}$ to $v_{i,\ss}$, all of which are labeled with $\CC^{[b]}$;\vspace{-0.06cm}
\item For every $b\in [g]$, add $\ell_b$ parallel edges from
  $v_{i,\ss}$ to $w_{i,b}$, all of which are labeled with $\CC^{[b]}$; and\vspace{-0.06cm}
\item Assign the vertex weight vector $\ww^{[k]}\in \frP$ to $v_{i,\ss}$.\vspace{0.06cm}
\end{enumerate}
It can be checked that the (undirected) subgraph spanned by $R_i$, for all $i\in [n]$,
  is connected.\vspace{0.012cm}

This almost finishes the construction.
The only thing left is to add some more vertices and edges
  so that the out-degree of $u_{i,a}$ and the in-degree of $w_{i,a}$
  are the same for all $i\in [n]$ and $a\in [g]$.\vspace{0.014cm}

To this end, we notice that for all $i\in [n]$ and $a\in [g]$,
  both the out-degree of $u_{i,a}$ and the in-degree of $w_{i,a}$ constructed so far
  are linear in the maximum degree of $G$, because all the parameters
  $k_i,\ell_i$ and the sets $\calS_i$ are considered as constants.
As a result, we can pick a large enough positive integer $M\ge 2$ which is
  linear in the maximum degree of $G$, such that
$$
M\ge\hspace{0.04cm} \text{the out-degree of $u_{i,a}$ and the in-degree of $w_{i,a}$
  constructed so far,\ for all $i$ and $a$.}
$$
We now add vertices and edges so that the out-degree of $u_{i,a}$
  and the in-degree of $w_{i,a}$ all become $M$.\vspace{0.012cm}

Let $i\in [n]$ and $a\in [g]$. Assume the current out-degree of $u_{i,a}$ is $k\le M$.
Then we add $M-k$ new Type-R vertices in $R_i$
  and add one edge from $u_{i,a}$ to each of these vertices.
The vertex weights of all the new vertices are $\mathbf{1}$,
  and the edge weights of all the new edges are $\DD^{[a]}$ (recall that
  we are allowed to use the normalized version $\DD^{[a]}$ of $\CC^{[a]}$, and this is actually the
  only place we use it).\vspace{0.015cm}

Similarly, assume the current in-degree of $w_{i,a}$ is $k\le M$.
Then we add $M-k$ new Type-L vertices in $R_i$
  and add one edge from each of these vertices to $w_{i,a}$.
The vertex weights of all the new vertices are $\mathbf{1}$ while
  the edge weights of all the new edges are $\CC^{[a]}$.\vspace{0.015cm}

This finishes the construction of the new labeled directed
  graph $\calG^{[1]}=(G^{[1]},\calV^{[1]},\calE^{[1]})$.

\subsection{Proof of Equation (\ref{lastbig})}

We start with the definition of $\text{nvw}(\xi)$, for any assignment $\xi:V=[n]\rightarrow [r]$.\vspace{0.01cm}

First, for each $a\in [g]$, we let $\bmu^{[a]}$ denote the following
  positive $r$-dimensional vector:
$$
\mu^{[a]}_i=\sum_{x\in A_i}\hspace{0.08cm} \Big(\alpha^{[1]}_x\Big)^2 \cdot \Big(\alpha^{[a]}_x\Big)^{M-2},
\ \ \ \ \ \text{for every $i\in [r]$.}
$$
For every $a\in [g]$, we let $\bnu^{[a]}$ denote the following
  positive $r$-dimensional vector:
$$
\nu^{[a]}_i=\sum_{x\in B_i}\hspace{0.08cm} \Big(\beta^{[1]}_x\Big)^2 \cdot \Big(\beta^{[a]}_x\Big)^{M-2},
\ \ \ \ \ \text{for every $i\in [r]$.}
$$
Finally, we define $\text{nvw}(\xi)$ as follows:\vspace{0.07cm}
$$
\text{nvw}(\xi)=\prod_{i\in [n]}\hspace{0.1cm} \prod_{a\in [g]}
  \hspace{0.08cm}\mu^{[a]}_{\xi(i)} \cdot \nu^{[a]}_{\xi(i)},
\ \ \ \ \ \text{for any $\xi:V=[n]\rightarrow [r]$.}
$$
It is easy to check that $\text{nvw}(\xi)>0$ and the number
  of possible values of $\text{nvw}(\xi)$ is polynomial in $n$.

Now we prove equation (\ref{lastbig}) for $k=1$:\vspace{0.04cm}
\begin{equation}\label{quququququ}
Z_{\frP,\frR}(\calG^{[1]})=
  \sum_{\xi:V \rightarrow [r]} \text{wt}(\calG,\xi) \cdot \text{nvw}(\xi).
\end{equation}

Let $\xi$ be an assignment from $V$ to $[r]$.
We use $\Phi_\xi$ to denote the set of all assignments $\phi:V^{[1]}\rightarrow [m]$ such that
  for every edge $uv$ in the subgraph spanned by $R_i$, $i\in [n]$,
  we have
$$
\phi(u)\in A_{\xi(i)}\ \ \ \ \text{and}\ \ \ \ \phi(v)\in B_{\xi(i)}.
$$
In other words, for all $i\in [n]$ and $v\in R_i$, if $v$ a Type-L vertex
  then $\phi(v)\in A_{\xi(i)}$; if $v$ is a Type-R vertex then $\phi(v)\in B_{\xi(i)}$;
  and if $v$ is a Type-M of $R_i$, then $\phi(v)\in A_{\xi(i)}\cap B_{\xi(i)}$.
Equivalently, we can associate every vertex $v\in V^{[1]}$ with a subset $U_v\subseteq [m]$, where\vspace{0.05cm}
\begin{enumerate}
\item If $v$ appears in both $R_i$ and $R_j$ for some $i\ne j\in V=[n]$,
  and $v$ is Type-R in $R_i$ and\\ Type-L in $R_j$, then $U_v=B_{\xi(i)}\cap A_{\xi(j)}$;
\item Otherwise, assume $v$ only appears in $R_i$ for some $i\in V=[n]$. Then
\begin{enumerate}
\item If $v$ is Type-L, then $U_v=A_{\xi(i)}$;
\item If $v$ is Type-R, then $U_v=B_{\xi(i)}$; and
\item If $v$ is Type-M, then $U_v=B_{\xi(i)} \cap A_{\xi(i)}$,\vspace{0.05cm}
\end{enumerate}\end{enumerate}
   such that
  $\phi\in \Phi_\xi$ if and only if $\phi(v)\in U_v$ for all $v\in V^{[1]}$.
In particular, $\Phi_\xi=\emptyset$ iff $U_v=\emptyset$ for some $v$.\vspace{0.005cm}

By the construction, we know the subgraph spanned by $R_i$ is \emph{connected}, for any $i\in [n]$.
It implies that $\text{wt}(\calG^{[1]},\phi)\ne 0$ only if
  $\phi\in \Phi_\xi$ for a unique $\xi:V\rightarrow [r]$.
As a result, we have
$$
Z_{\frP,\frR}(\calG^{[1]})=\sum_{\phi}\hspace{0.05cm} \text{wt}(\calG^{[1]},\phi)
=\sum_{\xi} \sum_{\phi\in \Phi_\xi} \text{wt}(\calG^{[1]},\phi),
$$
and to prove (\ref{quququququ}) we only need to show that
$$
\sum_{\phi\in \Phi_\xi} \text{wt}(\calG^{[1]},\phi)=\text{wt}(\calG,\xi)\cdot
  \text{nvw}(\xi),\ \ \ \ \ \text{for any
  assignment $\xi:V=[n]\rightarrow [r]$.}
$$

We use $\ww^{[v]}$ to denote the weight of $v\in V^{[1]}$, $E_i$ to denote
  the set of edges in $E^{[1]}$ labeled with $\CC^{[i]}$, and $F_i$ to
  denote the set of edges in $E^{[1]}$ labeled with $\DD^{[i]}$, then we have\vspace{0.1cm}
$$
\sum_{\phi\in \Phi_\xi} \text{wt}(\calG^{[1]},\phi)\hspace{0.04cm}
  =\sum_{\phi\in \Phi_\xi} \left(\prod_{v\in V^{[1]}} \ww^{[v]}_{\phi(v)}
  \hspace{0.05cm}\prod_{i\in [g]}\left( \prod_{uv\in E_i} C^{[i]}_{\phi(u),\phi(v)}\right)
  \left(\prod_{uv\in F_i}D^{[i]}_{\phi(u),\phi(v)}\right)\right).
$$\newpage
\noindent By the definition of $\Phi_\xi$, if $\Phi_\xi\ne \emptyset$, then every $\phi\in \Phi_\xi$
  satisfies\vspace{0.05cm}
$$
C^{[i]}_{\phi(u),\phi(v)}=\alpha^{[i]}_{\phi(u)}\cdot \beta^{[i]}_{\phi(v)}\ \ \ \ \text{and}
\ \ \ \ D^{[i]}_{\phi(u),\phi(v)}=\alpha^{[i]}_{\phi(u)}\cdot \delta^{[i]}_{\phi(v)},\vspace{0.05cm}
$$
where $(\balpha^{[i]},\bdelta^{[i]})$ is the representation of $\DD^{[i]}$.
As a result, we have\vspace{0.06cm}
$$
\sum_{\phi\in \Phi_\xi} \text{wt}(\calG^{[1]},\phi)
  =\sum_{\phi\in \Phi_\xi} \left(\prod_{v\in V^{[1]}} \ww^{[v]}_{\phi(v)}
  \hspace{0.05cm}\prod_{i\in [g]}\left( \prod_{uv\in E_i} \alpha^{[i]}_{\phi(u)}\cdot
  \beta^{[i]}_{\phi(v)}\right)\left(\prod_{uv\in F_i} \alpha^{[i]}_{\phi(u)}
  \cdot \delta^{[i]}_{\phi(v)}\right)\right), \vspace{0.1cm}
$$
Because $\phi\in \Phi_\xi$ iff $\phi(v)\in U_v$ for all $v$,
  we can express this sum of products as a product of sums:
$$
\prod_{v\in V^{[1]}} H_v,
$$
in which every $H_v$, $v\in V^{[1]}$, is a sum over $\phi(v)\in U_v$.

Finally, we show the following equation:
\begin{equation}\label{qusieq}
\prod_{v\in V^{[1]}} H_v=\text{wt}(\calG,\xi)\cdot \text{nvw}(\xi).
\end{equation}
This follows from the construction of $\calG^{[1]}$ and the following observations:\vspace{0.07cm}
\begin{enumerate}
\item For each $v_e\in R_i\cap R_j$, which is added because of edge $ij\in E$,
  it can be checked that the sum $H_{v_e}$\\ over $U_{v_e}=B_{\xi(i)}\cap A_{\xi(j)}$
  is exactly $F_{\xi(i),\xi(j)}$,
  where $\FF$ is the weight of $ij$ in $\calG$ (as defined in (\ref{lateruse})).

\item Let $\pp$ denote the vertex weight of $i\in V$, which is generated using
  $\calS_1,\calS_2$ and $\calS_3$.
Then we have
$$
p_{\xi(i)}=\prod_{\ss\in \calS_1} H_{v_{i,\ss}} \prod_{\ss\in \calS_2}
  H_{v_{i,\ss}} \prod_{\ss\in \calS_3} H_{v_{i,\ss}}.
$$

\item For all $i\in [n]$ and $a\in [g]$, we have\vspace{-0.05cm}
$$
\mu^{[a]}_{\xi(i)}=H_{u_{i,a}}\ \ \ \ \text{and}\ \ \ \
\nu^{[a]}_{\xi(i)}=H_{w_{i,a}}.\vspace{-0.06cm}
$$

\item Finally, it can be checked that $H_v=1$ for all other vertices in $V^{[1]}$,
  which is the reason we need to\\ use the normalized matrices $\DD^{[a]}$ in the construction.\vspace{0.05cm}
\end{enumerate}

\subsubsection{Construction of $\calG^{[k]}$}

We can similarly construct $\calG^{[k]}$ for every $k\in [L]$.

The only difference is that, instead of $u_{i,a}$ and $w_{i,a}$,
  we add the following $2kg$ vertices in $R_i$:
$$
u_{i,j,a}\ \ \text{and}\ \ w_{i,j,a},\ \ \ \ \ \text{for all $j\in [k]$ and $a\in [g]$.}
$$
We also connect these vertices by adding $4\hspace{0.01cm}k\hspace{0.01cm}g$ edges,
  whose underlying undirected graph is a cycle.
All these edges are labeled with $\CC^{[1]}$.
We also add extra vertices and edges so that the out-degree
  of $u_{i,j,a}$ and the in-degree of $v_{i,j,a}$ are $M$ for all $i\in [n]$, $j\in [k]$
  and $a\in [g]$.
It then can be proved similarly that\vspace{0.03cm}
$$
Z_{\frP,\frR}(\calG^{[k]})= \sum_{\xi:V\rightarrow [r]}
  \text{wt}(\calG,\xi)\cdot \Big(\text{nvw}(\xi)\Big)^k.
$$
This completes the proof of Lemma \ref{reductionlemm}.

\section{Decidability}\label{decidability}

In this section, we show that the rank condition is decidable
  in a finite number of steps.

\subsection{A Technical Lemma}

We prove a very useful technical lemma.

\begin{lemm}\label{huhulemma}
Let $L,n,m\ge 1$ be positive integers.
For every $i\in [L]$, let $\{a^{[i]}_1,\ldots,a^{[i]}_n\}$ be a
  sequence of $n$ positive numbers; and
let $\{b^{[i]}_1,\ldots,b^{[i]}_m\}$ be a sequence of $m$
  positive numbers. If\vspace{0.05cm}
$$
\sum_{i\in [n]}\hspace{0.07cm} \prod_{j\in [L]}\Big(a^{[j]}_i\Big)^{k_j}=
\sum_{i\in [m]}\hspace{0.07cm} \prod_{j\in [L]}\Big(b^{[j]}_i\Big)^{k_j},
\ \ \ \ \ \text{for all $k_1,k_2,\ldots,k_L\ge 1$,}\vspace{0.04cm}
$$
then $m=n$ and there exists a one-to-one correspondence $\pi$ from $[n]$ to itself
  such that
$$
a^{[j]}_i=b^{[j]}_{\pi(i)},\ \ \ \ \ \text{for all $i\in [n]$ and $j\in [L]$.}
$$
\end{lemm}
\begin{proof}
We prove it by induction on $L$. The base case when $L=1$ is trivial.

Assume the lemma is true for $L-1\ge 1$. Without loss of generality, we assume that
  $\{a^{[L]}_1,\ldots,a^{[L]}_n\}$ and $\{b^{[L]}_1,\ldots,b^{[L]}_m\}$
  are already sorted:
$$
a^{[L]}_1\ge \ldots\ge a^{[L]}_n>0\ \ \ \ \ \text{and}\ \ \ \ \
b^{[L]}_1\ge \ldots\ge b^{[L]}_m>0.
$$
\noindent We let $s\ge 1$ and $t\ge 1$ be the two maximum integers such that
$$
a^{[L]}_1=a^{[L]}_2=\cdots =a^{[L]}_s=a>0\ \ \ \ \ \text{and}\ \ \ \ \
b^{[L]}_1=b^{[L]}_2=\cdots =b^{[L]}_t=b>0.
$$

First it is easy to show that $a=b$. Otherwise assume $a>b$, then we set
  $k_1=\ldots=k_{L-1}=1$, divide $(a)^{k_L}$ from both sides, and let $k_L$ go to infinity.
It is easy to check that the left side converges to
$$
\sum_{i\in [s]}\hspace{0.07cm} \prod_{j\in [L-1]} a^{[j]}_i>0,
$$
while the right side converges to $0$, which contradicts the assumption.\vspace{0.006cm}

Second, we fix $k_1,\ldots,k_{L-1}$ to be any positive integers, divide
  $(a)^{k_L}=(b)^{k_L}$ from both sides and let $k_L$ go to infinity.
It is easy to check that the left side converges to\vspace{-0.1cm}
$$
\sum_{i\in [s]}\hspace{0.07cm} \prod_{j\in [L-1]} \Big(a^{[j]}_i\Big)^{k_j},
$$
while the right hand side converges to\vspace{-0.2cm}
$$
\sum_{i\in [t]}\hspace{0.07cm} \prod_{j\in [L-1]} \Big(b^{[j]}_i\Big)^{k_j}.
$$
So these two sums are equal for all $k_1,\ldots,k_{L-1}\ge 1$.
Then we apply the inductive hypothesis to claim that $s=t$ and there exists
  a permutation $\pi$ from $[s]$ to itself such that
\begin{equation}\label{eq55}
a^{[j]}_i=b^{[j]}_{\pi(i)},\ \ \ \ \ \text{for all $j\in [L-1]$ and $i\in [s]$.}
\end{equation}
It is also easy to see that for any $i\in [s]$, (\ref{eq55}) also holds for $j=L$.

We then repeat the whole process after removing the first $s$ elements from
  the $2L$ sequences.\vspace{0.05cm}
\end{proof}

Additionally, we also need the following simple lemma in the proof.\vspace{0.05cm}

\begin{lemm}\label{forklore}
Let $m\ge 1$ be an integer and $(P_1, P_2, \ldots, )$ be
  a sequence of subsets of $[m]$. If for any finite subset
  $\{i_1,\ldots,i_k\}\subset \mathbb{N}$,
$
P_{i_1}\cap P_{i_2}\cap \cdots \cap P_{i_k}\ne \emptyset,
$
then there exists a $j\in [m]$ such that $j\in P_i$ for all $i$.\vspace{0.025cm}
\end{lemm}
\begin{proof}
If for every $j \in [m]$, there exists some $i_j\ge 1$ such that
  $j \not \in P_{i_j}$, then the finite intersection
$$\bigcap_{j=1}^{m} P_{i_j} = \emptyset,$$
which contradicts the assumption.
\end{proof}

\subsection{Matrix and Vector Polynomials}

Let $(\frS,\frT)$ be a generalized $\calP$-pair, for some $m\times m$ pattern $\calP$.
So every vector $\ww\in \frS$ is either \emph{positive} or \emph{$\calP$-weakly positive}
  and every $\DD\in \frT$ is either a \emph{$\calP$-matrix} or a \emph{$\calP$-diagonal matrix}.
Note that if $\frT$ only has $\calP$-matrices, then $(\frS,\frT)$ is a $\calP$-pair.
The definitions below also apply to $\calP$-pairs.\vspace{0.008cm}

We say $f$ is a \emph{$\calP$-matrix polynomial} if $f$ is a polynomial over variables
$$
\Big\{\hspace{0.05cm}x_{i,j}:(i,j)\in \calP \Big\}
$$
with integer coefficients and zero constant term.
We say $\frT$ satisfies $f$ if for every $\calP$-matrix $\DD\in \frT$, we have $f(\DD)=0$, in which
  we substitute $x_{i,j}$ by $D_{i,j}>0$ for all $(i,j)\in \calP$.
We also say $(\frS,\frT)$ satisfies $f$ if $\frT$ satisfies $f$.\vspace{0.008cm}

We say $f$ is a \emph{$\calP$-diagonal {matrix polynomial}} if $f$ is a polynomial over variables
$$
\Big\{\hspace{0.05cm}x_{i}:(i,i)\in \calP\Big\}
$$
with integer coefficients and zero constant term.
We say $\frT$ satisfies $f$ if every $\calP$-diagonal matrix $\DD\in \frT$ satisfies $f(\DD)=0$.
We also say $(\frS,\frT)$ satisfies $f$ if $\frT$ satisfies $f$.\vspace{0.01cm}

We say $g$ is an \emph{$m$-vector polynomial} if $g$ is a polynomial over variables
$$
\Big\{\hspace{0.05cm}y_i:i\in [m]\Big\}
$$
with integer coefficients and zero constant term.
Similarly, we say $\frS$ satisfies $g$ if every positive vector $\ww\in \frS$ satisfies $g(\ww)=0$.
We also say $(\frS,\frT)$ satisfies $g$ if $\frS$ satisfies $g$.\vspace{0.007cm}

Finally, we say $g$ is a \emph{$\calP$-weakly positive vector polynomial} if $g$ is a polynomial over variables
$$
\Big\{\hspace{0.05cm}y_i:(i,i)\in \calP\Big\}
$$
with integer coefficients and zero constant term.
We say $\frS$ satisfies $g$ if every
  $\calP$-weakly positive vector $\ww\in \frS$ satisfies $g(\ww)=0$.
We also say $(\frS,\frT)$ satisfies $g$ if $\frS$ satisfies $g$.\vspace{0.008cm}

Let $F$ be a finite set of $\calP$-matrix, $\calP$-diagonal matrix,
  $m$-vector, and $\calP$-weakly positive vector polynomials.
Then we say $(\frS,\frT)$ satisfies $F$ if $(\frS,\frT)$ satisfies
  every polynomial $f\in F$.\vspace{0.008cm}

Similarly, given any block pattern $\calT$,
  we can define $\calT$-matrix polynomials, $\calT$-diagonal matrix polynomials,
  and $\calT$-weakly positive vector polynomials for $\calT$-pairs
  and \emph{generalized} $\calT$-pairs.

We remark that, for the case when $(\frS,\frT)$ is a $\calT$-pair, to check whether $\frT$
  satisfies the rank condition (i.e., every matrix $\DD\in \frT$ is block-rank-$1$),
  one only needs to check whether $\frT$ satisfies all the $\calT$-matrix polynomials
  $f_{i,i',j,j'}$ of the following form
$$
f_{i,i',j,j'}(\xx)=x_{i,j}\cdot x_{i',j'}-x_{i,j'}\cdot x_{i',j},\ \ \ \ \ \text{where
  $i,i'\in A_k$ and $j,j'\in B_k$ for some $k\in [r]$}.\vspace{0.1cm}
$$

\subsection{Checking Matrix and Vector Polynomials}

Now let $(\frS,\frT)$ be a $\calT$-pair for some non-trivial $m\times m$ block
  pattern $\calT=\{(A_1,B_1),\ldots,(A_r,B_r)\}$ with $r\ge 1$.
We also assume that every matrix in $\frT$ is block-rank-$1$,
  and $\frS$ is \emph{closed}.\vspace{0.007cm}

We can apply the \genp\ operation to get a new $\calP$-pair
$$
(\frS',\frT')=\genp(\frS,\frT),\ \ \ \ \ \text{where $\calP=\gen(\calT)$.}
$$
We also let $(\frS^*,\frT^*)$ denote the \emph{generalized} $\calP$-pair defined
  in Appendix \ref{construction}.
By definition, $\frS^*$ is also closed.\vspace{0.006cm}

In this section, we first show that to check whether $(\frS^*,\frT^*)$ satisfies
  a matrix or vector polynomial, one only needs to check
  finitely many polynomials for $(\frS,\frT)$.
One can prove a similar relation between $(\frS',\frT')$ and
  $(\frS^*,\frT^*)$.
As a result, to check whether $(\frS',\frT')$ satisfies a polynomial or not,
  we only need to check finitely many polynomials for $(\frS,\frT)$.\vspace{0.008cm}

We start with the following lemma.\vspace{0.05cm}

\begin{lemm}\label{matrixp1}
Let $f$ be a $\calP$-matrix or $\calP$-diagonal matrix polynomial.
Then one can construct a finite
  set $\{F_1,\ldots,F_L\}$ in a finite number of steps, in which every $F_i$, $i\in [L]$, is
  a finite set of $\calT$-matrix, $m$-vector, and $\calT$-weakly positive vector
  polynomials, such that
$$
\text{$(\frS^*,\frT^*)$ satisfies $f$}\ \ \Longleftrightarrow\ \
  \exists\hspace{0.06cm}i\in [L]\ \forall\hspace{0.02cm}{g \in F_{i}},
  \ \big[\text{$(\frS,\frT)$ satisfies $g$}\big].
$$
\end{lemm}

\begin{proof}
We first prove the case when $f$ is a $\calP$-matrix polynomial.\vspace{0.004cm}

If $f$ is the zero polynomial, then the lemma follows by setting $L=1$ and
  $F_1$ to be the set consists~of the zero polynomial only.
From now on we assume that $f$ is not the zero polynomial.\vspace{0.008cm}

Let $\{\CC^{[1]},\ldots,\CC^{[s]}\}$ and $\{\DD^{[1]},\ldots,\DD^{[t]}\}$ be two
  finite subsets of $\calT$-matrices in $\frT$
  and $\{\ww^{[1]},\ldots,\ww^{[h]}\}$ be a finite subset of \emph{positive}
  vectors in $\frS$, where $s,t,h\ge 1$.
We also let $(\balpha^{[i]},\bbeta^{[i]})$ and $(\bgamma^{[i]},\bdelta^{[i]})$
  denote the representations of $\CC^{[i]}$ and $\DD^{[i]}$, respectively.
By the definition of $\frT^*$ and the assumption that $\frT$ is \emph{closed},
  we can construct from every $(s+t+h)$-tuple
$$
\pp=\big(k_1,\ldots,k_s,\ell_1,\ldots,\ell_t,
  e_1,\ldots, e_h\big),\ \ \ \ \ \text{where $k_i,\ell_i,e_i\ge 1$,}
$$
the following $\calP$-matrix $\CC^{[\pp]}$ in $\frT^*$:
  the $(i,j)$th entry of $\CC^{[\pp]}$ is\vspace{0.06cm}
\begin{equation}\label{haha}
\sum_{x\in B_i\cap A_j} \Big(\beta^{[1]}_x\Big)^{k_1}
  \cdots \Big(\beta^{[s]}_x\Big)^{k_s}\cdot \Big(\gamma^{[1]}_x\Big)^{\ell_1}
  \cdots \Big(\gamma^{[t]}_x\Big)^{\ell_t}\cdot \Big(w^{[1]}_x\Big)^{e_1}
  \cdots \Big(w^{[h]}_x\Big)^{e_h},\ \ \ \ \ \text{for all $i,j\in [r]$.}
\end{equation}
This follows from the fact that the Hadamard product of $(\ww^{[1]})^{e_1},\ldots,
  (\ww^{[h]})^{e_h}$ is actually a vector in $\frS$,
  because $\frS$ is known to be \emph{closed}.\vspace{0.008cm}

Now we assume $(\frS^*,\frT^*)$ satisfies $f$, then by definition we must have\vspace{-0.04cm}
\begin{equation}\label{hahaa}
f(\CC^{[\pp]})=0,\ \ \ \ \ \text{for all $\pp$,}\vspace{-0.04cm}
\end{equation}
since $\CC^{[\pp]}$ is a $\calP$-matrix in $\frT^*$.
By combining (\ref{hahaa}) and (\ref{haha}) and \emph{rearranging} terms, we have\vspace{0.1cm}
\begin{eqnarray*}
&&\sum_{i\in [n_1]} \left( \prod_{j\in [s]} \Big(f_{i}\big(\beta^{[j]}_1,
  \ldots,\beta^{[j]}_m\big)\Big)^{k_j}\right)
  \left(\prod_{j\in [t]} \Big(f_{i}\big(\gamma^{[j]}_1,\ldots,\gamma^{[j]}_m\big)\Big)^{\ell_j}\right)
  \left(\prod_{j\in [h]} \Big(f_{i}\big(w^{[j]}_1,\ldots,w^{[j]}_m\big)\Big)^{e_j}\right)
  \\[2ex]&&=
\sum_{i\in [n_2]} \left( \prod_{j\in [s]} \Big(g_{i}\big(\beta^{[j]}_1,\ldots,\beta^{[j]}_m
\big)\Big)^{k_j}\right)
  \left(\prod_{j\in [t]} \Big(g_{i}\big(\gamma^{[j]}_1,\ldots,\gamma^{[j]}_m\big)\Big)^{\ell_j}\right)
  \left(\prod_{j\in [h]} \Big(g_{i}\big(w^{[j]}_1,\ldots,w^{[j]}_m\big)\Big)^{e_j}\right) \\[-1ex]
\end{eqnarray*}
for all $\pp$.
In the equation above, $n_1$ and $n_2$ are two non-negative integers.
For all $i\in [n_1]$ and $j\in [n_2]$, both $f_i(x_1,\ldots,x_m)$
  and $g_j(x_1,\ldots,x_m)$ are monomials in $x_1,\ldots,x_m$.
Also note that all the monomials $f_i,g_j$ only depend on the $\calP$-matrix polynomial
  $f$ but do not depend on the choices of $\pp$ and the
  subsets $\{\CC^{[1]},\ldots,\CC^{[s]}\}$, $\{\DD^{[1]},\ldots,\DD^{[t]}\}$,
  and $\{\ww^{[1]},\ldots,\ww^{[h]}\}$.
Moreover, because we assumed that $f$ is not the zero polynomial,
  at least one of $n_1$ and $n_2$ is nonzero.
  \vspace{0.007cm}


It follows directly from Lemma \ref{huhulemma} that if $(\frS^*,\frT^*)$ satisfies
  $f$, then we must have $n_1=n_2$ which~we denote by $n$.
(If $n_1\ne n_2$, then we already know that $f(\CC^{[\pp]})=0$ cannot hold
  for all $\pp$.
The lemma then follows by setting $L=1$ and $F_1$ to be the set consisting of
  the following $m$-vector polynomial: \hspace{-0.05cm}$g(\xx)=x_1$ so that $(\frS,\frT)$ does not satisfy $F_1$.)\vspace{0.003cm}
Moreover, by Lemma \ref{huhulemma}, if $(\frS^*,\frT^*)$ satisfies $f$ then there also exists a
  permutation $\pi$ from $[n]$ to itself such that
\begin{eqnarray*}
f_i\big(\beta^{[j]}_1,\ldots,\beta^{[j]}_m\big)=g_{\pi(i)}\big(\beta^{[j]}_1,\ldots,\beta^{[j]}_m\big),
&&\ \text{for all $j\in [s]$ and $i\in [n]$;}\\[1.7ex]
f_i\big(\gamma^{[j]}_1,\ldots,\gamma^{[j]}_m\big)=g_{\pi(i)}\big(\gamma^{[j]}_1,\ldots,\gamma^{[j]}_m\big),
&&\ \text{for all $j\in [t]$ and $i\in [n]$; and}\\[1.7ex]
f_i\big(w^{[j]}_1,\ldots,w^{[j]}_m\big)=g_{\pi(i)}\big(w^{[j]}_1,\ldots,w^{[j]}_m\big),
&&\ \text{for all $j\in [h]$ and $i\in [n]$.}\vspace{0.06cm}
\end{eqnarray*}

Since all the discussion above and all the monomials $f_i,g_i$
  do not depend on the choice of the three subsets,
  we can apply Lemma \ref{forklore} to claim that if $(\frS^*,\frT^*)$
  satisfies $f$, then there must exist a (universal) permutation
  $\pi$ from $[n]$ to itself such that for all $\DD\in \frS$ (since
  $(\frS,\frT)$ is a $\calT$-pair, $\DD$ is a $\calT$-matrix),
\begin{eqnarray*}
f_i(\alpha_1,\ldots,\alpha_m)-g_{\pi(i)}(\alpha_1,\ldots,\alpha_m)=0,&&\ \ \text{for all $i\in [n]$ and}\\[1.2ex]
f_i(\beta_1,\ldots,\beta_m)-g_{\pi(i)}(\beta_1,\ldots,\beta_m)=0,&&\ \ \text{for all $i\in [n]$,}
\end{eqnarray*}
where $(\balpha,\bbeta)$ is the representation of $\DD$; and
for every positive vector $\ww\in \frT$,
$$
f_i(w_1,\ldots,w_m)-g_{\pi(i)}(w_1,\ldots,w_m)=0,\ \ \ \ \ \ \text{for all $i\in [n]$.}
$$
It is also easy to check that these conditions are sufficient.

Furthermore, $\balpha$ and $\bbeta$ can be expressed by the positive entries of $\DD$ as follows.
For every $i\in A_k$, where $k\in [r]$, let $d$ be the smallest index in $B_k$, then we have
$$
\alpha_i=\frac{D_{i,d}}{\sum_{j\in A_k} D_{j,d}}.
$$
For every $i\in B_k$, where $k\in [r]$, let $d$ be the smallest index in $A_k$, then $\beta_i= D_{d,i}/\alpha_d$.
Now it is easy to see that for every permutation $\pi$ from $[n]$ to itself, we can construct
  a finite set $F_\pi$ of~$\calT$-matrix and $m$-vector polynomials, such that,
if $(\frS^*,\frT^*)$ satisfies $f$ then $(\frS,\frT)$ satisfies $F_\pi$ for some $\pi$.\vspace{0.003cm}

The case when $f$ is a $\calP$-diagonal matrix polynomial can be proved similarly.
The only difference is that every $F_\pi$ is now a
  finite set of $\calT$-matrix and $\calT$-weakly positive vector polynomials.\vspace{0.06cm}
\end{proof}

It also follows directly by definition that $\frT'$ satisfies a $\calP$-matrix polynomial if and only if
  $\frT^*$ satisfies the same polynomial, because $\frT'$ contains precisely all the $\calP$-matrices
  in $\frT^*$.
Next, we deal with vector polynomials.\vspace{0.06cm}

\begin{lemm}\label{matrixp3}
Let $g$ be an $r$-vector or a $\calP$-weakly positive vector polynomial.
One can construct a finite set $\{G_1,\ldots,G_L\}$ in a finite number
  of steps, in which every $G_i$ is
  a finite set of $\calT$-matrix, $m$-vector, and $\calT$-weakly positive vector polynomials, such that
$$\text{$(\frS^*,\frT^*)$ satisfies $g$}\ \ \Longleftrightarrow
  \ \ \exists\hspace{0.06cm}i\in [L]\ \forall\hspace{0.02cm}{f \in G_{i}},
  \ \big[\text{$(\frS,\frT)$ satisfies $f$}\big].
$$
\end{lemm}
\begin{proof}
We only prove the case when $g$ is $\calP$-weakly positive.
The other case can be proved similarly.\vspace{0.004cm}

Again, we assume that $g$ is not the zero polynomial.\vspace{0.01cm}

Recall that when defining $\frS^*$ in Appendix \ref{construction}, we first define $\frS^\#$ and
  $\frS^*$ is then the closure of $\frS^\#$:
$\ww$ is a $\calP$-weakly positive vector in $\frS^*$ if and only if there exist a finite and possibly empty
  subset of positive vectors $\{\ww^{[1]},\ldots,\ww^{[s]}\}\subseteq \frS^\#$ for some $s\ge 0$,
  a finite and nonempty subset of $\calP$-weakly positive vectors $\{\uu^{[1]},\ldots,\uu^{[t]}\}\subseteq \frS^\#$ for some $t\ge 1$, and positive integers $k_1,\ldots,k_s,\ell_1,\ldots,\ell_t$, such that
$$
\ww=\big(\ww^{[1]}\big)^{k_1}\circ\cdots\circ \big(\ww^{[s]}\big)^{k_s}\circ \big(\uu^{[1]}\big)^{\ell_1}
  \circ\ldots\circ \big(\uu^{[t]}\big)^{\ell_t}.
$$

To prove Lemma \ref{matrixp3}, we first construct a finite set $\{F_1,\ldots,F_M\}$,
  in which every $F_i$ is a finite set of
  $r$-vector and $\calP$-weakly positive vector polynomials, such that
\begin{equation}\label{ggh}
\text{$\frS^*$ satisfies $g$}\ \ \Longleftrightarrow\ \
  \exists\hspace{0.06cm}i\in [M]\ \forall\hspace{0.02cm}{f \in F_{i}},
  \ \big[\text{$\frS^\#$ satisfies $f$}\big].
\end{equation}
To this end, we let $\{\ww^{[1]},\ldots,\ww^{[s]}\}$ be a finite subset of positive vectors in
  $\frS^\#$; and $\{\uu^{[1]},\ldots,\uu^{[t]}\}$ be a finite subset of
  $\calP$-weakly positive vectors in $\frS^{\#}$, with $s\ge 0$ and $t\ge 1$. Then from any tuple
$$
\pp=\big(k_1,\ldots,k_s,\ell_1,\ldots,\ell_t\big),\ \ \ \ \ \ \text{where $k_i,\ell_i\ge 1$,}
$$
we get a $\calP$-weakly positive vector $\ww^{[\pp]}\in \frS^*$, where\vspace{0.02cm}
$$
\ww^{[\pp]}=\big(\ww^{[1]}\big)^{k_1}\circ\cdots\circ \big(\ww^{[s]}\big)^{k_s}
  \circ \big(\uu^{[1]}\big)^{\ell_1}\circ\cdots\circ \big(\uu^{[t]}\big)^{\ell_t}.
$$
Assume $\frS^*$ satisfies $g$, then we have $g(\ww^{[\pp]})=0$ for all $\pp$.
Combining these two equations, we have\vspace{0.14cm}
\begin{eqnarray*}
\sum_{i\in [n_1]} \left( \prod_{j\in [s]} \Big(f_{i}\big(\ww^{[j]}\big)\Big)^{k_j}\right)
  \left(\prod_{j\in [t]} \Big(f_{i}\big(\uu^{[j]}\big)\Big)^{\ell_j}\right) =
  \sum_{i\in [n_2]} \left( \prod_{j\in [s]} \Big(g_{i}\big(\ww^{[j]}\big)\Big)^{k_j}\right)
  \left(\prod_{j\in [t]} \Big(g_{i}\big(\uu^{[j]}\big)\Big)^{\ell_j}\right)\\[-1ex]
\end{eqnarray*}
for all $\pp$. In the equation, $f_i(\xx)$ and
  $g_i(\xx)$ are both monomials over $x_i$, $(i,i)\in \calP$.
Again, $f_i$ and $g_i$ only depend on the polynomial $g$ but do not depend on
  the choices of $\pp$ and the two subsets $\{\ww^{[1]},\ldots,\ww^{[s]}\}$ and
  $\{\uu^{[1]},\ldots,\uu^{[t]}\}$.

Because $g$ is not the zero polynomial, one of $n_1$ and $n_2$ must be positive,
  and we have the following two cases.
If $n_1\ne n_2$, then by Lemma \ref{huhulemma}, $\frS^*$ cannot satisfy
  $g$ and (\ref{ggh}) follows by setting $L=1$ and $F_1$ to be the set
  consists of the following $r$-vector polynomial: $f(\xx)=x_1$.\vspace{0.005cm}

Otherwise, we have $n_1=n_2>0$, which we denote by $n$.
It follows from Lemma \ref{huhulemma} and Lemma \ref{forklore} that
  if $\frS^*$ satisfies $g$, then there exists a universal permutation $\pi$ from $[n]$
  to itself such that for every positive and $\calP$-weakly positive vector $\ww\in \frS^\#$,
$$
f_i(\ww)=g_{\pi(i)}(\ww),\ \ \ \ \ \ \text{for all $i\in [n]$}.
$$
As a result, we can construct $F_\pi$ for each $\pi$, and $\frS^*$ satisfies $g$ if and only if
  $\frS^\#$ satisfies $F_\pi$ for some $\pi$.\vspace{0.01cm}

In the second step, we show that for any $r$-vector or $\calP$-weakly positive vector
  polynomial $f$, one can construct $\{F_1,\ldots,F_L\}$ in a finite number of steps, in which each
  $F_i$ is a finite set of $\calT$-matrix, $m$-vector and $\calT$-weakly positive vector polynomials,
such that, $\frS^\#$ satisfies $f$ if and only if $(\frS,\frT)$ satisfies
  $F_i$ for some $i\in [L]$.
The idea of the proof is very similar to the proof of Lemma \ref{matrixp1}
  so we omit it here.\vspace{0.015cm}

Lemma \ref{matrixp3}, for the case when $g$ is $\calP$-weakly positive,
  then follows by combing these two steps.\vspace{0.04cm}
\end{proof}\newpage

We can also prove the following lemma similarly.\vspace{0.05cm}

\begin{lemm}\label{matrixp5}
Let $g$ be an $r$-vector or a $\calP$-weakly positive vector polynomial.
Then one can construct a finite set $\{G_1,$ $\ldots,G_L\}$ in a finite number of steps, in which every $G_i$,
  $i\in [L]$, is a finite set of $\calP$-matrix, $\calP$-diagonal matrix,
  $r$-vector, and $\calP$-weakly positive vector polynomials, such that
$$\text{$(\frS',\frT')$ satisfies $g$}\ \ \Longleftrightarrow\ \
\exists\hspace{0.06cm}i\in [L]\ \forall\hspace{0.02cm}{f \in G_{i}},
  \ \big[\text{$(\frS^*,\frT^*)$ satisfies $f$}\big].$$
\end{lemm}

\subsection{Decidability of the Rank Condition}

Finally, we use these lemmas to prove Lemma \ref{decidabilitylemma},
  the decidability of the rank condition.\vspace{0.005cm}

We start with the following simple observation.
Let $F=\{f_1,\ldots,f_s\}$ be a finite set of matrix and vector polynomials.
For each $i\in [s]$, there is a finite set $\{F_{i,1}, \ldots,F_{i,L_i}\}$
  in which every $F_{i,j}$ is some finite set of polynomials,
  and we have the following statement:
$$
\text{$(\frS',\frT')$ satisfies $f_i$}\ \ \Longleftrightarrow\ \
  \exists\hspace{0.06cm}j\in [L_i]\ \forall\hspace{0.02cm}{f \in F_{i,j}},
  \ \big[\text{$(\frS,\frT)$ satisfies $f$}\big].
$$
Then the conjunction of these statements over $f_i\in F$, $i\in [s]$, can be expressed in the same form:
One can construct from $\{F_{i,j}:i\in [s],j\in [L_i]\}$ a new finite set $\{G_1,\ldots, G_L\}$ in which
  every $G_j$ is some finite set of polynomials, such that
$$
\forall\hspace{0.02cm}{f \in F},\ \big[\text{$(\frS',\frT')$ satisfies $f$}\big]
\ \ \Longleftrightarrow\ \
\exists\hspace{0.06cm}j\in [L]\ \forall\hspace{0.02cm}{g \in G_j},\ \big[\text{$(\frS,\frT)$ satisfies $g$}\big].$$

Now we prove Lemma \ref{decidabilitylemma}.
After $\ell\ge 0$ steps, we get a sequence of $\ell+1$ pairs
$$
(\frS_0,\frT_0),(\frS_1,\frT_1),\ldots,(\frS_\ell,\frT_\ell),
$$
which satisfies condition \textbf{(R$_\ell$)}.
Since we assumed that $\frS_0=\{\11\}$, every $\frS_i$ in the
  sequence is \emph{closed}.\vspace{0.005cm}

We show how to check whether every matrix $\DD\in \frT_{\ell+1}$,
  where
$$
(\frS_{\ell+1},\frT_{\ell+1})=\genp(\frS_{\ell},\frT_{\ell}),
$$
is block-rank-$1$ or not.
To this end we first check whether $\calP=\gen(\calT_\ell)$
  is consistent with a block pattern or not.
If not, then we conclude that $\frT_{\ell+1}$ does not
  satisfy the rank condition.\vspace{0.005cm}

Otherwise, we use $\calT_{\ell+1}$ to denote the block pattern
  consistent with $\calP$.
To check the rank condition, it is equivalent to
  check whether $\frT_{\ell+1}$ satisfies the following
  $\calP$-matrix polynomials:
$$
f_{i,i',j,j'}(\xx)=x_{i,j}\cdot x_{i',j'}-x_{i,j'}\cdot x_{i',j},\ \ \ \ \ \text{where
  $i,i'\in A_k$ and $j,j'\in B_k$ for some $k\in [r]$}
$$
and $(A_1,B_1),\ldots,(A_r,B_r)$ are the pairs in $\calT_{\ell+1}$.\vspace{0.006cm}

By Lemma \ref{matrixp1}-\ref{matrixp5}, we can construct a finite set $\{F_1,\ldots,F_L\}$
  in which every $F_i$ is a finite set of
\begin{center}
$\calT_{\ell}$-matrix, $m_\ell$-vector, and $\calT_{\ell}$-weakly positive vector
  polynomials
\end{center}such that
\begin{center}
$\frT_{\ell+1}$ satisfies the rank condition if and only if
  $(\frS_{\ell},\frT_{\ell})$ satisfies $F_i$ for some $i\in [L]$.
\end{center}
\newpage \noindent If $\ell=0$, then we are done, since $(\frS_0,\frT_0)$ is finite and
  we can check all the polynomials in $F_i$ for all $i\in [L]$ in a finite number of steps.
Otherwise, $\ell\ge 1$ and we can use Lemma \ref{matrixp1}--\ref{matrixp5} and the observation above
  to construct, for each $F_i$, a finite set $\{F_{i,1},\ldots,F_{i,L_i}\}$ in which
  every $F_{i,j}$ is a finite set of
\begin{center}
$\calT_{\ell-1}$-matrix,
  $m_{\ell-1}$-vector, and $\calT_{\ell-1}$-weakly positive vector polynomials
\end{center}such that
\begin{center}
$(\frS_{\ell},\frT_{\ell})$ satisfies
  $F_i$ if and only if $(\frS_{\ell-1},\frT_{\ell-1})$ satisfies $F_{i,j}$
  for some $j\in [L_i]$.
\end{center}

We repeat this process until we reach the finite pair $(\frS_0,\frT_0)$.
So the checking procedure looks like a~huge tree of depth $\ell+1$.
Every leaf $v$ of the tree is associated with a finite set $F_v$ of\vspace{0.06cm}
\begin{center}
$\calT_0$-matrix, $m_0$-vector, and $\calT_0$-weakly positive vector polynomials.\vspace{0.06cm}
\end{center}
Set $\frT_{\ell+1}$ satisfies the rank condition if and only if
  $(\frS_0,\frT_0)$ satisfies $F_v$ for some leaf $v$ of the tree.

\section{The Dichotomy for the $\{0,1\}$ Case}\label{Appen-cri}

We briefly describe the dichotomy criterion of Bulatov \cite{BulatovECCC}.

A finite {\it relational structure} ${\cal H}$
 over a finite set of
 relational symbols $R_1, R_2, \ldots, R_k$, each of which has
a fixed arity, is a  non-empty set $H$ together with an interpretation
of these relational symbols $R_1^{\cal H},$ $R_2^{\cal H}, \ldots,
 R_k^{\cal H}$ which are relations on $H$ of the corresponding
arities. For graph homomorphism (i.e., ${\cal H}$-coloring),
we start with a single binary
relation, namely the edge relation $E$ on $H$.
A relation $R$ is said to be {\it pp-definable} in ${\cal H}$,
if it can be expressed by the relations $R_i^{\cal H}$, $1 \le i \le k$,
together with  the binary {\sc Equality} predicate on ${\cal H}$,
 conjunction, and existential quantifiers.

A mapping $f$ from $H^m$ to $H$, for some $m \ge 1$,
is called a {\it polymorphism} of ${\cal H}$ if it satisfies the following
condition: For any relation $R \in \{R_1^{\cal H}, R_2^{\cal H}, \ldots,
R_k^{\cal H}\}$ of arity $n$, for any  $m$ tuples in $H^n$:
$$(a_{1,1}, \ldots, a_{1,n}), \ldots, (a_{m,1}, \ldots, a_{m,n}) \in
H^n,$$ if each
$(a_{i,1}, \ldots, a_{i,n}) \in R$ for all $i:1 \le i \le m$, then
$$\big(f(a_{1,1}, \ldots, a_{m,1}), \ldots, f(a_{1,n}, \ldots, a_{m,n})\big) \in R.$$
A relational structure ${\cal H}$ defines a {\it universal algebra}
${\bf A}$, where the universe is $H$ and the set of all
polymorphisms are its operations.
A theorem of Geiger then
states that
a relation on $H$ is invariant under all
polymorphisms iff it is pp-definable.

A pp-definable binary equivalence relation is called a {\it congruence}.
A subalgebra is a unary pp-definable relation (subset)
together with the restrictions of the given relations.
One can easily define direct product algebras and
homomorphic images (quotient algebra modulo a congruence).
%
%
A class of universal algebras closed under quotient, subalgebra
and direct product is called a {\it variety}.
The class of algebras that~are homomorphic images
of subalgebras of direct powers of some universal algebra
is called the variety gene\-rated by it (HSP theorem).

A Mal'tsev polymorphism $m$ is a ternary polymorphism satisfying
$m(x, x, y) = y$ and $m(x, y, y) = x$, for all $x, y \in H$.
Having a Mal'tsev polymorphism is a necessary condition for tractability.

Now start with the relational structure ${\cal H}$
with a single edge relation $E$, then add to it all the
unary relations $\{C_h \mid h \in H\}$, where $C_h = \{(h)\}$, we obtain
 a relational structure denoted by  ${\cal H}_{\rm id}$.
Then the polymorphisms of ${\cal H}_{\rm id}$ define
the universal algebra called the \emph{full idempotent reduct}.
These are the idempotent polymorphisms of ${\cal H}$:
$f(x, \ldots, x) = x$.\newpage

Congruences form a lattice. Given any two
congruences $\alpha$ and $\beta$, we
let $A_1, \ldots, A_s$ and $B_1, \ldots, B_t$ be
the equivalence classes of $\alpha$ and $\beta$ respectively,
then  the $s \times t$ matrix $M(\alpha, \beta)$
has $(i,j)$ entry $|A_i \cap B_j|$.

The tractability criterion of Bulatov can now be stated:
Start with ${\cal H}_{\rm id}$ and take the full idempotent reduct.
The \#CSP problem defined by ${\cal H}$
is tractable iff every finite algebra ${\bf A}$
 in the {\it variety} generated by this
full idempotent reduct satisfies the following condition:
For any two
congruences $\alpha$ and $\beta$
in  ${\bf A}$, the ${\rm rank} (M(\alpha, \beta))$
is equal to the number of equivalence classes of $\alpha
\vee \beta$, the join congruence of $\alpha$ and $\beta$.


The reason it is difficult to show that this dichotomy criterion is
  {decidable} is because it talks about \emph{all} finite algebras  ${\bf A}$
in the {\it variety} generated by the full idempotent reduct
of ${\cal H}_{\rm id}$.  This variety is infinite, containing
arbitrarily large arities over $H$.
Thus, even though in graph homomorphism we are given only a binary
relation, the process of forming the variety produces
arbitrarily large arities,
and this criterion is a condition involving infinitely many relations.

\end{document}